\documentclass[aps,%
twocolumn,%
floatfix,%
pra,%
amsfonts,%
groupedaddress,%
superscriptaddress]{revtex4-2}%
\usepackage[utf8]{inputenc}
\usepackage{mathrsfs}       
\usepackage{mathtools}
\usepackage{amsmath}
\usepackage{amsfonts}       
\usepackage{amssymb}
\usepackage[titletoc,title]{appendix}
\usepackage[ruled,linesnumbered]{algorithm2e}    
\usepackage{bbm}            
\usepackage{bm}             
\usepackage{comment}
\usepackage{amsthm}   
\usepackage{booktabs} 
\usepackage[T1]{fontenc}
\usepackage[osf,sc]{mathpazo}
\usepackage{dsfont}
\usepackage{microtype}
\usepackage{verbatim}
\usepackage{pifont}
\newcommand{\cmark}{\ding{51}}%
\newcommand{\xmark}{\ding{55}}%

\allowdisplaybreaks

\usepackage{thmtools}
\usepackage{makecell}
\usepackage{array}
\usepackage{colortbl}
\usepackage[dvipsnames]{xcolor}
\usepackage{tikz}
\usetikzlibrary{arrows.meta,calc,positioning,shapes.geometric}
\definecolor{beamer@blendedblue}{rgb}{0.2,0.2,0.7}  
\usepackage[colorlinks=true,%
linkcolor=beamer@blendedblue,%
bookmarks=true,%
breaklinks=true,%
filecolor=beamer@blendedblue,%
anchorcolor=yellow,%
citecolor=beamer@blendedblue,%
urlcolor=beamer@blendedblue,%
pdfauthor={Kun Wang and Masahito Hayashi},%
pdfsubject={Efficient Verification of Entangled Measurements with Local States},%
CJKbookmarks=true]{hyperref}

\usepackage[skins,breakable,most]{tcolorbox}

\usepackage{float}    

\newcommand{\sbar}{\;\rule{0pt}{9.5pt}\middle|\;}

\newtheorem{definition}{Definition}
\newtheorem{theorem}{Theorem}
\newtheorem{proposition}[theorem]{Proposition}

\newtheorem{lemma}[theorem]{Lemma}
\newtheorem{corollary}[theorem]{Corollary}

\mathchardef\ordinarycolon\mathcode`\:
\mathcode`\:=\string"8000
\def\vcentcolon{\mathrel{\mathop\ordinarycolon}}
\begingroup \catcode`\:=\active
  \lowercase{\endgroup
  \let :\vcentcolon
  }

\tcolorboxenvironment{proof}{%
  enhanced,
  boxrule=0.5pt,
  breakable,
  sharp corners,
  before upper={%
    \setlength{\parindent}{2em}%
    \setlength{\parskip}{0pt}%
  }%
}

\DeclareFontFamily{U}{mathx}{\hyphenchar\font45}
\DeclareFontShape{U}{mathx}{m}{n}{<-> mathx10}{}
\DeclareSymbolFont{mathx}{U}{mathx}{m}{n}
\DeclareMathAccent{\widebar}{0}{mathx}{"73}

\newcommand{\wh}[1]{\widehat{#1}}

\newcommand{\ket}[1]{\vert{#1}\rangle}
\newcommand{\bra}[1]{\langle{#1}\vert}
\newcommand{\ketbra}[2]{\vert{#1}\rangle\!\langle{#2}\vert}

\newcommand\proj[1]{\vert{#1}\rangle\!\langle{#1}\vert}

\DeclareMathOperator{\tr}{Tr}  

\newcommand{\1}{\mathbbm{1}}
\newcommand{\ox}{\otimes}

\def\sss{\scriptscriptstyle}
\newcommand{\Fid}{F} 

\makeatletter
\newsavebox{\@brx}
\newcommand{\llangle}[1][]{\savebox{\@brx}{\(\m@th{#1\langle}\)}%
  \mathopen{\copy\@brx\kern-0.5\wd\@brx\usebox{\@brx}}}
\newcommand{\rrangle}[1][]{\savebox{\@brx}{\(\m@th{#1\rangle}\)}%
  \mathclose{\copy\@brx\kern-0.5\wd\@brx\usebox{\@brx}}}
\makeatother

\newcommand*{\bF}{\mathbb{F}}

\newcommand*{\cB}{\mathcal{B}}

\newcommand*{\cH}{\mathcal{H}}

\newcommand*{\cJ}{\mathcal{J}}

\newcommand*{\cM}{\mathcal{M}}

\newcommand*{\cO}{\mathcal{O}}
\newcommand*{\cP}{\mathcal{P}}
\newcommand*{\cQ}{\mathcal{Q}}

\newcommand*{\cS}{\mathcal{S}}

\newcommand*{\cX}{\mathcal{X}}

\newcommand*{\bZ}{\mathbb{Z}}

\definecolor{wildstrawberry}{rgb}{1.0, 0.26, 0.64}
\definecolor{googleblue}{HTML}{4285F4}
\definecolor{googlered}{HTML}{DB4437}
\definecolor{googleyellow}{HTML}{F4B400}
\definecolor{googlegreen}{HTML}{0F9D58}

\begin{document}

\newcommand{\thetitle}{{Efficient Verification of Entangled Measurements with Local States}}
\title{\thetitle}

\author{Kun Wang}
\email{kunwang.quantum@outlook.com}
\affiliation{Independent Researcher}

\author{Masahito Hayashi}
\email{hmasahito@cuhk.edu.cn}
\affiliation{School of Data Science, The Chinese University of Hong Kong,
Shenzhen, Longgang District, Shenzhen, 518172, China}
\affiliation{International Quantum Academy (SIQA), Futian District, Shenzhen 518048, China}

\date{\today}
\begin{abstract}
We develop a framework for quantum measurement verification (QMV) using only local state preparations. For locally transitive and irreducible projective measurements, we prove that symmetry reduces locality constrained QMV to quantum state verification of a single basis state, thereby reducing protocol design to the optimization of homogeneous verification operators. We apply the framework to generalized Bell measurements, single-parameter measurements on two qubits, elegant joint measurements, and stabilizer state induced measurements, and derive explicit local protocols together with closed form verification operators, success probabilities, and sample complexities. We further show that homogeneous QMV protocols can estimate measurement fidelity directly from observed passing frequencies.
\end{abstract}
\maketitle

\section{Introduction}\label{sec:introduction}

Measurement devices are the operational interface between quantum systems and
classical decisions.  Rather than merely reading out a computation at the end,
their outcomes trigger feedback, syndrome decoding, postselection, network
connections, and acceptance tests.  Consequently, reliable measurements are
indispensable for paradigmatic tasks such as quantum teleportation, dense
coding, entanglement swapping, quantum repeaters, quantum communication, and
fault-tolerant computation~\cite{bennett1993teleporting,gisin2019entanglement}.
At the same time, measurement remains a noisy hardware operation.  Recent
large-scale demonstrations across superconducting, trapped ion, neutral atom,
and photonic quantum computing platforms all benchmark nonzero readout or
state preparation and measurement errors
~\cite{GoogleQuantumAI2025BelowThreshold,Liu2025CertifiedRandomness,Bluvstein2024Logical,PsiQuantum2025Manufacturable}.
An unverified measurement device can therefore invalidate the classical output
of an otherwise controlled quantum experiment.  It is thus essential to
verify the correct functioning of measurement devices before using them in
practical applications.  Full detector tomography is typically time intensive
and computationally difficult~\cite{PhysRevA.64.024102,Lundeen2008}, motivating
more efficient methods that certify whether a device measures input states as
designed.

In this work we propose a general framework for quantum measurement verification (QMV)
based on \emph{local state preparations}, thereby establishing a remaining piece
of quantum verification theory alongside quantum state
verification~\cite{Hayashi_2006,H1,Hayashi_2009,Hayashi2015Verifiable,Morimae2017Verification,Hayashi2018Self,pallister2018optimal,Takeuchi2018Verification,Wang2019Optimal,Yu2019Optimal,Li2019Efficient,Zhu2019General,Liu2019Efficient,Zhu2019Optimal,zhu2019efficient,Li2019Efficient,Zhu2019EfficientPRL,Li2020Optimal,Dangniam2020Optimal,Zhang2020Experimental,Jiang2020Towards,Zhang2020Classical,Li2021Verification,Han2021Optimal,Liu2021Efficient,Wu2021Efficient,Liu2021Universally,Li2021Minimum,MiguelRamiro2022Collective,Liu2022Verification,Zhu2022Efficient,Chen2022Efficient,Xia2022Experimental,Gocanin2022Sample},
quantum process verification~\cite{liu2020efficient,Zhu2020Efficient,Zeng2020Quantum,Luo2022Proof,Zhang2022Efficient}, and quantum subspace verification~\cite{Zheng2026Efficient,Chen2025QuantumSubspace,Zheng2025GHZ,Hu2026Nearly}.
We refer the reader to~\cite{yu2022statistical} for a comprehensive review on this topic.
Specifically, we demonstrate efficient local verification protocols for
various measurements that are of practical interest, including
generalized Bell measurements,
single-parameter measurements on two qubits,
elegant joint measurements, and
stabilizer state induced measurements on $n$ qudits of prime local dimension.
These protocols are constructed purely from local product states and
achieve the characteristic verification scaling $\cO(1/\varepsilon)$, in contrast
to the $\cO(1/\varepsilon^2)$ scaling of direct fidelity estimation.
Notably, the requirement that the test states be local incurs only a constant factor penalty
over the globally optimal protocol without constraints.

\textbf{Organization.}
Section~\ref{sec:preliminaries} reviews quantum state verification and quantum measurements.
Section~\ref{sec:qmv} formulates QMV and develops the symmetry reduction under locality constraints.
Section~\ref{sec:applications} applies the general result to representative classes of measurements.
Section~\ref{sec:measurement-fidelity-estimation} discusses measurement fidelity estimation 
based on homogeneous QMV protocols.

\section{Preliminaries}\label{sec:preliminaries}

\subsection{Quantum state verification}

Here we review the quantum state verification (QSV)
task~\cite{pallister2018optimal,yu2022statistical}.
A quantum device is designed to produce a multipartite pure state $\ket{\psi}$.
However, it might work incorrectly and outputs
independent states $\sigma_1, \sigma_2, \ldots, \sigma_N$ in $N$ runs.
It is guaranteed that either $\sigma_j=\proj{\psi}$ for all $j$ (good case)
or $\bra{\psi}\sigma_j\ket{\psi}\leq 1-\varepsilon$ for all $j$ (bad case).
We aim to identify the case with the worst-case failure probability $\delta$.
To achieve this, a verifier performs binary measurements randomly chosen
from a set of available measurements.
Each measurement $\{T_\ell,\1-T_\ell\}$ is specified by some operator $T_\ell$
and is performed with probability $p_\ell$, corresponding to ``passing'' the test. 
Let $\Omega:=\sum_\ell p_\ell T_\ell$ denote the corresponding verification strategy. In the bad case, 
the single-round worst-case passing probability that $\sigma_j$ passes the
test satisfies~\cite{pallister2018optimal}
\begin{align}\label{eq:pallister2018optimal}
P(\Omega,\varepsilon) :=  
\max_{\bra{\psi}\sigma_j\ket{\psi}\leq1-\varepsilon}\tr[\Omega\sigma_j]
= 1 - [1-\lambda_2(\Omega)]\varepsilon,
\end{align}
where $\lambda_2(\Omega)$ is the second largest eigenvalue of $\Omega$.
Correspondingly, all $N$ quantum states in the bad case can pass the test
with probability at most $[1 - (1-\lambda_2(\Omega))\varepsilon]^N$.
To achieve the failure probability $\delta$, it suffices to take
\begin{align}\label{eq:state-verification-sample-complexity}
N(\Omega)
&=\left\lceil
\frac{\ln\delta}{\ln[1-[1-\lambda_2(\Omega)]\varepsilon]}
\right\rceil,
\end{align}
where $\ln$ denotes the natural logarithm.
We call $N(\Omega)$ the sample complexity of $\Omega$ by abuse of notation.

\subsection{Quantum measurements}

The most general kind of measurements in quantum mechanics is called the
\emph{positive-operator-valued measures (POVMs)}.
A POVM is a set of operators $\{M_\theta\}_{\theta\in\Theta}$ satisfying $\forall\theta\in\Theta, M_\theta\geq0$
and $\sum_{\theta\in\Theta}M_\theta=\1$, where $\Theta$ is a finite alphabet,
$\theta\in\Theta$ records the measurement outcome, and $\1$ is the identity operator~\cite{Hayashi2017}.
Performing the measurement on a quantum state $\rho$,
the probability of outcome $\theta$ is given by the Born rule: $p_\theta = \tr[M_\theta\rho]$.
If every operator $M_\theta$ of the POVM is a rank-one projector,
then the corresponding measurement is also called a \emph{von Neumann measurement}.

An equivalent characterization of quantum measurements from the quantum channel perspective
is to regard the resulting states as purely classical states.
Mathematically, the induced measurement channel of a POVM $\{M_\theta\}_{\theta\in\Theta}$
is defined as~\cite{Hayashi2017}
\begin{align}\label{eq:measurement-channel}
  \cM(\rho) = \sum_{\theta\in\Theta}\tr[M_\theta\rho]\proj{\theta},
\end{align}
where $\{\ket{\theta}\}_{\theta\in\Theta}$ is the computational basis of the underlying Hilbert space.
Thus, a measurement channel takes a quantum system to a classical one.

The QSV task reviewed above is closely related to QMV, but it cannot be applied
directly.  In QSV, the untrusted object is a state source; in QMV, it is an
entire measurement apparatus, described by a faulty POVM
$\cM=\{M_\theta\}_{\theta\in\Theta}$ whose elements are jointly constrained by
$\sum_\theta M_\theta=\1$.  Hence the adversary may redistribute probability
among outcomes, and the worst-case passing probability is not the expectation
value of one state against one verification operator.
Verifying the Choi state of the measurement channel in Eq.~\eqref{eq:measurement-channel} 
would require process verification, while directly testing a von Neumann measurement
$\cP=\{\proj{\psi_\theta}\}_{\theta\in\Theta}$ by preparing
$\ket{\psi_\theta}$ is often unavailable when these states are entangled.  These
obstacles motivate a dedicated QMV framework with local input states, which we
now formulate.

\begin{figure*}[!t]
    \centering
    \resizebox{0.96\textwidth}{!}{%
    \begin{tikzpicture}[
      x=1cm,y=1cm,
      >=Latex,
      every node/.style={font=\footnotesize},
      box/.style={rounded corners=5pt, line width=0.7pt},
      panel/.style={box, draw=#1!65!black, fill=#1!6},
      action/.style={box, draw=#1!65!black, fill=#1!10,
        minimum width=1.75cm, minimum height=0.55cm, align=center},
      question/.style={ellipse, draw=googlegreen!60!black,
        fill=googlegreen!9, line width=0.75pt, align=center, inner sep=1.5pt,
        font=\scriptsize, text width=1.55cm, minimum width=1.85cm,
        minimum height=0.88cm},
      arrow/.style={-{Latex[length=2.2mm]}, line width=0.85pt, black!72}
    ]
      \node[box, draw=googleblue!65!black, fill=white, align=center,
        font=\small, minimum width=5.25cm, minimum height=0.58cm] (prot) at (8.20,4.3)
        {$\cS=\{p_x,(\rho_x,\cP(\rho_x))\}_{x\in\cX}$};
      \node[box, draw=googleyellow!75!black, fill=googleyellow!15,
        font=\small, align=center, minimum width=2.10cm, minimum height=0.58cm] (coin)
        at (8.20,3.50) {sample $x\sim p_x$};
      \draw[arrow] (prot.south) -- (coin.north);

      \draw[panel=googleblue] (0.00,-0.30) rectangle (4.05,2.85);
      \draw[panel=black] (5.05,-0.30) rectangle (9.15,2.85);
      \draw[panel=googlegreen] (10.45,-0.30) rectangle (16.40,2.85);

      \node[font=\bfseries, text=googleblue!55!black] at (2.025,2.56)
        {state preparation};
      \node[font=\bfseries, text=black!70] at (7.10,2.56)
        {measurement device};
      \node[font=\bfseries, text=googlegreen!45!black] at (13.425,2.56)
        {verification};

      \node[action=googleblue, minimum width=2.45cm, minimum height=0.88cm]
        (prep) at (2.025,1.34)
        {};
      \node[font=\Large] at ($(prep.center)+(0,0.14)$) {$\rho_x$};
      \node[font=\scriptsize] at ($(prep.center)+(0,-0.22)$) {test state};

      \node[box, draw=black!70, fill=black!4, align=center,
        minimum width=2.45cm, minimum height=0.88cm] (dev) at (7.10,1.34)
        {};
      \node[font=\Large] at ($(dev.center)+(0,0.14)$) {$\mathscr{D}$};
      \node[font=\scriptsize] at ($(dev.center)+(0,-0.22)$) {actual POVM $\cM$};

      \node[question,font=\small] (ask)
        at (11.825,1.34) {$\theta\in\cP(\rho_x)$?};
      \node[question,font=\small] (rounds)
        at (15.025,1.34) {$\#\geq N$?};
      \node[box, draw=googlegreen!55!black, fill=googlegreen!9,
        text=googlegreen!45!black, font=\bfseries, align=center,
        minimum width=0.80cm, minimum height=0.45cm] (bad) at (11.825,0.05)
        {Bad};
      \node[box, draw=googlered!70!black, fill=googlered!8,
        text=googlered!70!black, font=\bfseries, align=center,
        minimum width=0.80cm, minimum height=0.45cm] (good) at (15.025,0.05)
        {Good};

      \draw[arrow] (prep.east) -- (dev.west);
      \draw[arrow] (dev.east) -- (ask.west);
      \draw[arrow, dashed, googleblue!65!black]
        (coin.west) -| (2.025,2.85);

      \draw[arrow, googlegreen!45!black] (ask.south) -- node[left,
        font=\normalsize, text=googlegreen!45!black] {\xmark} (bad.north);
      \coordinate (yesmid) at ($(ask.east)!0.5!(rounds.west)$);
      \draw[arrow, googlegreen!45!black] (ask.east) -- (rounds.west);
      \node[font=\normalsize, text=googlegreen!45!black]
        at ($(yesmid)+(0,0.16)$) {\cmark};
      \draw[arrow, googlegreen!45!black] (rounds.south) -- node[right,
        font=\normalsize, text=googlegreen!45!black] {\cmark} (good.north);
      \draw[arrow, dashed, googlegreen!45!black] (rounds.north) --
        node[pos=0.18, right, font=\normalsize, text=googlegreen!45!black] {\xmark}
        (15.025,3.50) -- (coin.east);
    \end{tikzpicture}}
    \caption{\textbf{Schematic view of quantum measurement verification.}
    In each round, the protocol $\cS$ samples an index $x$ with probability $p_x$,
    which determines both the test state $\rho_x$ and the accept set $\cP(\rho_x)$.
    The device measures $\rho_x$ and returns an outcome $\theta$.
    If $\theta\notin\cP(\rho_x)$, verification outputs \textbf{Bad}.
    If $\theta\in\cP(\rho_x)$ and the prescribed total number of rounds $N$ has been reached, it outputs \textbf{Good};
    otherwise, it proceeds to the next test round.}
    \label{fig:qmv}
\end{figure*}
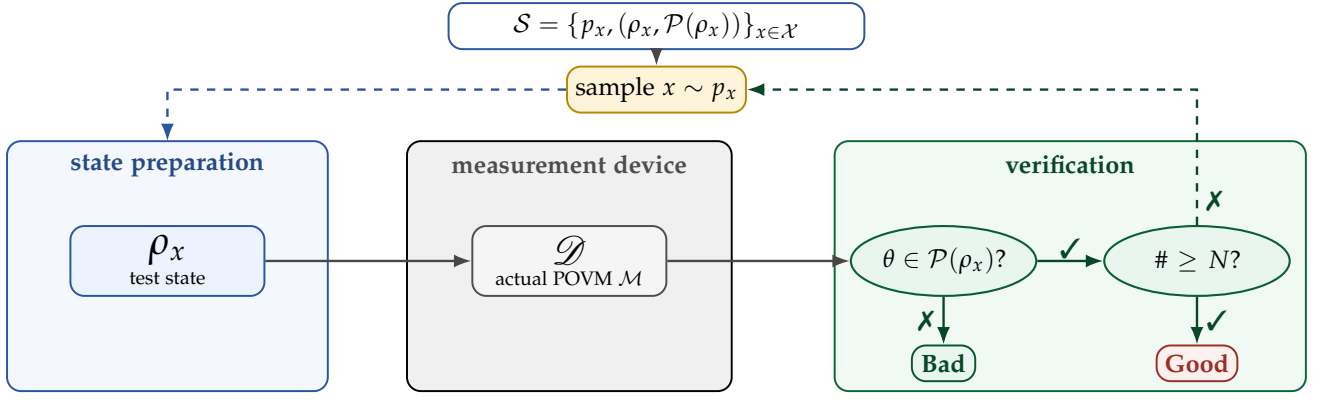

\section{Quantum Measurement Verification}\label{sec:qmv}

We first establish a general framework for quantum measurement verification.  
Then, by leveraging the group symmetry within von Neumann measurements, 
we show when the QMV problem constrained by locality can be reduced to a 
QSV problem for one representative basis state, 
where existing tools facilitate the design of efficient and even optimal verification strategies.

\subsection{General framework}

Consider a quantum measurement device $\mathscr{D}$ that is supposed to implement a
von Neumann measurement on $n$ parties $\cP=\{\proj{\psi_\theta}\}_{\theta\in\Theta}$,
where $\{\ket{\psi_\theta}\}_{\theta\in\Theta}$ forms an orthonormal basis of $\cH^{\otimes n}$
and $\cH$ is a Hilbert space of dimension $d$.
The total number of measurement outcomes is $\vert\Theta\vert=d^n$.
However, $\mathscr{D}$ might not work correctly and actually
realizes a POVM $\cM=\{M_\theta\}_{\theta\in\Theta}$ that is different from $\cP$.
In order to distinguish these two cases by making a reasonable number of uses of $\mathscr{D}$,
we are promised that the noisy measurement $\cM$ is sufficiently far from the target $\cP$.
That is, we need to distinguish between the following two cases with the worst-case failure probability $\delta$:
\begin{description}
  \item[\textbf{Good}] $\mathscr{D}$ accurately implements $\cP$;
  \item[\textbf{Bad}] $\mathscr{D}$ implements $\cM$, such that $\Fid(\cP,\cM) \leq 1-\varepsilon$, where
        $\Fid(\cP,\cM) := \frac{1}{\vert\Theta\vert}\sum_{\theta\in\Theta}\bra{\psi_\theta}M_\theta\ket{\psi_\theta}$
        is the measurement fidelity quantifying the closeness between $\cM$ and $\cP$ and $\varepsilon\in(0,1)$ is
        the precision.
\end{description}

The ideal measurement $\cP$ naturally assigns to each input state a
set of outcomes that may occur with nonzero probability.  More precisely,
for an input state $\rho$, define
\begin{align}
  \cP(\rho)
  :=
  \left\{
  \theta\in\Theta \;\middle|\;
  \bra{\psi_\theta}\rho\ket{\psi_\theta}>0
  \right\}.
\end{align}
This set is the support of the outcome distribution generated by the
ideal measurement $\cP$ on the input $\rho$, and it will be used as the
accept set for that input.

Once the accept set is fixed in this way by the target measurement, 
a QMV protocol is specified by
\begin{align}
  \cS=\{p_x,(\rho_x,\cP(\rho_x))\}_{x\in\cX},
\end{align}
where $\rho_x$ is chosen with probability $p_x$
and $\cP(\rho_x)$ denotes the associated accept set of measurement outcomes.
A single round of the protocol proceeds as follows. 
The verifier samples $x$ according to
$p_x$, prepares $\rho_x$, and feeds it into the device $\mathscr D$.
If the returned outcome $\theta$ belongs to $\cP(\rho_x)$, the device
passes this round; otherwise it fails.  
The rounds are repeated until either a failure occurs or 
the prescribed number $N$ of rounds has been reached.
If any round fails, the verifier stops and declares that
$\mathscr D$ is in the \textbf{Bad} case.  
If all $N$ rounds pass, the verifier declares that $\mathscr D$ is in the \textbf{Good} case. 
The verification procedure is illustrated in Figure~\ref{fig:qmv}.
The choice of $\cP(\rho_x)$ enjoys two crucial advantages:
1) it guarantees that if the device \emph{is} in the \textbf{Good} case, it always passes the test; and
2) it is optimal in the sense that any strictly larger set than $\cP(\rho_x)$ can only increase
the chance of being fooled by an adversary~\cite{pallister2018optimal}.

Given a protocol $\cS$, we associate with each outcome $\theta\in\Theta$
the operator
\begin{align}
  \Omega_\theta(\cS) := \sum_{x\in\cX:\theta\in\cP(\rho_x)} p_x\rho_x.
  \label{eq:outcome-verification-operator}
\end{align}
We call $\Omega_\theta(\cS)$ the verification operator associated with
outcome $\theta$.  It collects, with their sampling probabilities, all
test states for which the ideal outcome $\theta$ is accepted.

Now fix an actual POVM $\cM=\{M_\theta\}_{\theta\in\Theta}$ implemented
by $\mathscr D$. The single-round passing probability of the protocol
$\cS$ against $\cM$ is
\begin{subequations}\label{eq:single-round-passing-probability}
\begin{align}
  p_{\rm pass}(\cS,\cM)
  &=
  \sum_{x\in\cX} p_x
  \sum_{\theta\in\cP(\rho_x)}
  \tr[\rho_x M_\theta] \\
  &=
  \sum_{\theta\in\Theta}
  \tr\!\left[\Omega_\theta(\cS)M_\theta\right].
\end{align}  
\end{subequations}
The second equality is just a rearrangement of the sum according to the
reported outcome.

The relevant figure of merit for QMV is the largest value of
this single-round passing probability over all POVMs in the
\textbf{Bad} case. 
We therefore define the single-round worst-case passing probability as
\begin{subequations}\label{eq:worst-case failure probability}
\begin{align}
  P(\cS,\varepsilon)
:=&\;
  \max_{\cM:\sss\Fid(\cP,\cM)\leq 1-\varepsilon}
  p_{\rm pass}(\cS,\cM) \\
=&\;
  \max_{\cM:\sss\Fid(\cP,\cM)\leq 1-\varepsilon}
  \sum_{\theta\in\Theta}
  \tr\!\left[\Omega_\theta(\cS)M_\theta\right].
\end{align}
\end{subequations}

We repeat the above verification round $N$ times independently.
An error occurs if the quantum device is in the \textbf{Bad} case
but we incorrectly decide that it is in the \textbf{Good} case.
The worst-case probability of this error is $P(\cS,\varepsilon)^N$,
since the rounds are carried out independently.
Thus, to achieve confidence $1-\delta$, it suffices to take
$P(\cS,\varepsilon)^N \leq \delta$, yielding the minimum number of tests
\begin{align}\label{eq:number-of-copies}
N(\cS,\varepsilon,\delta) = \left\lceil\frac{\ln\delta}{\ln P(\cS,\varepsilon)}\right\rceil,
\end{align}
assuming that $P(\cS,\varepsilon)<1$.
Our task is to construct good verification protocols that minimize $P(\cS,\varepsilon)$.

For completeness, Appendix~\ref{appx:sec:properties} collects several
basic properties of the maximal passing probability
$P(\cS,\varepsilon)$.  These observations clarify the behavior of the
figure of merit under refinement of test states, variation of
$\varepsilon$, mixing of protocols, and local equivalence of target
measurements.

\subsection{Optimal verification without constraints}

If there were no constraints on the set of test states,
the optimal QMV protocol has been proposed in~\cite{liu2020efficient}
and works as follows. We randomly prepare an input state $\ket{\psi_\theta}$
with probability $1/\vert\Theta\vert$ and measure it with $\mathscr{D}$.
If the measurement outcome is $\theta$, $\mathscr{D}$ passes the test;
otherwise $\mathscr{D}$ fails the test. Namely,
\begin{align}
  \cS_{\rm opt} = \left\{1/\vert\Theta\vert, (\proj{\psi_\theta}, \{\theta\})\right\}_{\theta\in\Theta}.
\end{align}
The maximal probability of this protocol reads
\begin{align}
  P(\cS_{\rm opt},\varepsilon)
= \max_{\sss\Fid(\cP,\cM)\leq 1-\varepsilon}
  \frac{1}{\vert\Theta\vert}\sum_{\theta\in\Theta}\tr[\psi_\theta M_\theta]
= 1-\varepsilon,
\end{align}
where $\psi_\theta:=\proj{\psi_\theta}$ and the upper bound follows from the definition of $\Fid(\cP,\cM)$ and is
tight by mixing the ideal measurement with a fixed point free relabeling of the
outcomes.
Then, Eq.~\eqref{eq:number-of-copies} implies that we can verify $\mathscr{D}$ with the number of tests
\begin{align}\label{eq:optimal-copies}
N_{\rm opt}
  = \left\lceil \frac{\ln\delta}{\ln(1-\varepsilon)}\right\rceil
  \approx
  \left\lceil
  \frac{1}{\varepsilon}\ln\frac{1}{\delta}
  \right\rceil,
\end{align}
where we have used the fact that $\lim_{x\to0}\ln(1+x)=x$.

The above protocol requires that we are able to prepare possibly highly
entangled states $\{\ket{\psi_\theta}\}_\theta$, which is experimentally challenging.
It is thus of practical interest to devise efficient or optimal QMV protocols
under reasonable physical constraints.
Nevertheless, it will serve as the benchmark against which we compare the
more constrained protocols investigated below.

\subsection{Symmetries simplify QMV with locality constraint}\label{sec:symmetries-simplify-QMV}

We now impose the locality constraint on the verification protocol: every test state
in $\cS$ must be a local product state. This constraint is essential for
experimental implementability, but it also makes the optimization in
Eq.~\eqref{eq:worst-case failure probability} difficult. The main point of this
section is that a natural symmetry of the target measurement removes most of
this difficulty: under suitable conditions, 
locality constrained QMV reduces to
locality constrained QSV for one representative basis state. 

\begin{definition}\label{def:locally-transitive}
Let $\cP=\{\proj{\psi_\theta}\}_{\theta\in\Theta}$ be a von Neumann measurement.
We say that $\cP$ is \textbf{locally transitive} if there exist a finite group
$G$ acting transitively on $\Theta$ and a local unitary representation $V$ of
$G$ on $\cH^{\otimes n}$ such that
\begin{align}\label{eq:A-symmetry}
  V(g)\proj{\psi_\theta}V^\dagger(g) = \proj{\psi_{g\theta}},\quad \forall g\in G,
\end{align}
where $V(g)=\bigotimes_{i=1}^n V_i(g)$ and each $V_i(g)$ acts on the local Hilbert space $\cH$.
\end{definition}

Given a local protocol $\cS=\{p_x,(\rho_x,\cP(\rho_x))\}_{x\in\cX}$, its group symmetrization is
\begin{align}
\bar{\cS}
=\left\{\frac{p_x}{|G|},
\left(V(g)\rho_xV^\dagger(g),\,g\cP(\rho_x)\right)
\right\}_{x\in\cX,g\in G},
\end{align}
where $g\cP(\rho_x):=\{g\theta:\theta\in\cP(\rho_x)\}$. Since $V(g)$ is local,
$\bar{\cS}$ is still a local protocol. We call $\cS$ \textbf{symmetric} if
this symmetrization leaves the protocol unchanged up to relabeling.
We prove the following in Appendix~\ref{appx:B1}.

\begin{proposition}\label{prop:symmetric-protocols}
If $\cP$ is locally transitive, then a local protocol $\cS$ can be replaced
by a symmetric local protocol $\bar{\cS}$ satisfying
\begin{align}\label{eq:symmetrization-improves}
P(\bar{\cS},\varepsilon)\leq P(\cS,\varepsilon).
\end{align}
Consequently, optimal local QMV protocols may be searched for within the class
of symmetric protocols.
\end{proposition}

The second symmetry condition concerns the stabilizer of one outcome.

\begin{definition}\label{def:irreducibility}
Let $\cP=\{\proj{\psi_\theta}\}_{\theta\in\Theta}$ be locally transitive with
respect to $G$. For each $\theta\in\Theta$, define the stabilizer
$H_\theta:=\{g\in G:g\theta=\theta\}$. We say that $\cP$ is
\textbf{irreducible} if the one dimensional space spanned by
$\ket{\psi_\theta}$ is an irreducible subspace of the representation of
$H_\theta$, and no irreducible subspace of its orthogonal complement is
equivalent to this one dimensional subspace. This condition is independent of
the choice of $\theta$.
\end{definition}

These two symmetry conditions play complementary roles. Local transitivity
allows any local protocol to be symmetrized without increasing the worst-case
passing probability, while irreducibility ensures that the symmetrized
verification operator for any one outcome has the same spectral structure as a
pure state verification strategy. Together, they reduce the design of efficient
local QMV protocols to the design of efficient local QSV protocols for a single
representative state $\ket{\psi_\theta}$. 
This reduction is formalized in the following theorem.

\begin{theorem}\label{thm:verification-efficiency}
Let $\cP=\{\proj{\psi_\theta}\}_{\theta\in\Theta}$ be a von Neumann
measurement that is both locally transitive and irreducible.  
Let $\cS$ be a symmetric local QMV protocol of $\cP$, and define
\begin{align}
  T_\theta(\cS):=|\Theta|\Omega_\theta(\cS),
  \qquad \theta\in\Theta .
\end{align}
The following statements hold:
\begin{enumerate}
\item For every $\theta\in\Theta$, $T_\theta(\cS)$ is a valid
verification operator for the state $\ket{\psi_\theta}$; namely,
\begin{align}
  0\leq T_\theta(\cS)\leq \1,
  \qquad
  T_\theta(\cS)\ket{\psi_\theta}=\ket{\psi_\theta}.
\end{align}
\item The QMV worst-case passing probability is equal to the QSV
worst-case passing probability for any fixed outcome $\theta$;
namely, for any $\theta\in\Theta$,
\begin{align}\label{eq:qmv-to-qsv}
  P(\cS,\varepsilon)
  =
  \max_{\substack{\sigma\geq0,\ \tr[\sigma]=1:\\
  \bra{\psi_{\theta}}\sigma\ket{\psi_{\theta}}\leq1-\varepsilon}}
  \tr[T_\theta(\cS)\sigma].
\end{align}
In particular, the value of the maximization on the right 
hand side is independent of the choice of $\theta$.
\item The second largest eigenvalue of $T_\theta(\cS)$ is independent
of $\theta$.  Denoting this common value by $\lambda_2(\cS)$, we have
\begin{align}
  P(\cS,\varepsilon)
  =
  1-\bigl[1-\lambda_2(\cS)\bigr]\varepsilon .
\end{align}
Equivalently, for every $\theta\in\Theta$,
\begin{align}
  \lambda_2(T_\theta(\cS))=\lambda_2(\cS).
\end{align}
\end{enumerate}
\end{theorem}

Theorem~\ref{thm:verification-efficiency} is the main reduction result and is shown in
Appendix~\ref{appx:sec:proof-of-theorem-verification-efficiency}.
It turns the original optimization over faulty POVMs in
Eq.~\eqref{eq:worst-case failure probability} into the familiar QSV problem
of minimizing the second eigenvalue of a single operator $T_\theta(\cS)$.
The representative outcome $\theta$ is arbitrary, so all outcomes are verified
with the same efficiency.

A particularly useful case is when the symmetric protocol is homogeneous:
\begin{align}\label{eq:homogeneous}
  \Omega_\theta(\cS)
  = \frac{\proj{\psi_\theta}+\beta(\1-\proj{\psi_\theta})}{d^n},
  \quad \forall \theta\in\Theta .
\end{align}
Then $T_\theta(\cS)=\proj{\psi_\theta}+\beta(\1-\proj{\psi_\theta})$ and
Theorem~\ref{thm:verification-efficiency} immediately gives
\begin{align}\label{eq:homogeneous-failure-probability}
P(\cS,\varepsilon)=1-(1-\beta)\varepsilon.
\end{align}
Accordingly, the number of tests required to achieve significance $\delta$ is
\begin{align}\label{eq:number-of-copies-2}
N(\cS,\varepsilon,\delta)
=\left\lceil
\frac{\ln\delta}{\ln[1-(1-\beta)\varepsilon]}
\right\rceil.
\end{align}
Thus, after symmetry reduction, designing an efficient local QMV protocol
amounts to constructing local test states that make $\lambda_2(T_\theta)$,
or equivalently $\beta$ in the homogeneous case, as small as possible.

\section{Applications}\label{sec:applications}

In this section, we apply the symmetry viewpoint of
Theorem~\ref{thm:verification-efficiency} to several representative measurements.
For each example, we first identify the local symmetry that
permutes the measurement outcomes, then construct a local protocol and display
its verification operators.  When both preconditions of
Theorem~\ref{thm:verification-efficiency} are satisfied, this reduces the
QMV problem constrained by locality to verification of a single representative
state; otherwise, we give a direct calculation of the homogeneous operator.

\subsection{Generalized Bell measurement}\label{sec:GBM}

The generalized Bell measurement (GBM) is the canonical joint measurement of two qu\textit{d}its
$\cH^{\otimes 2}\equiv\cH_A\otimes \cH_B$ onto a maximally entangled basis.
It is a fundamental primitive in quantum information theory, most notably in quantum
teleportation~\cite{bennett1993teleporting}. 
We use the symmetry reduction of Theorem~\ref{thm:verification-efficiency} to obtain the 
optimal local state verification protocol for the GBM.

Let $\bZ_d:=\{0,\cdots, d-1\}$.
The canonical maximally entangled state in a Hilbert space of dimension $d^2$,
$\cH^{\otimes 2}$, is formally defined as
\begin{align}\label{eq:entangled-qudit-mes}
  \vert\Psi^\star\rangle := \frac{1}{\sqrt{d}}\sum_{j\in\bZ_d}\ket{jj},
\end{align}
where $\{\ket{j}:j\in\bZ_d\}$ is the computational basis of $\cH$.
Let $\{W_{x,z}\}_{(x,z)\in\bZ_d^2}$ be the set of discrete Weyl operators on $\cH$;
see Appendix~\ref{appx:sec:GBM} for the detailed definition.
The GBM on $\cH^{\otimes 2}$ is defined as
\begin{align}\label{eq:GBM}
  \cB = \left\{\vert\Psi_{x,z}\rangle\!\langle\Psi_{x,z}\vert : x,z\in \bZ_d\right\},
\end{align}
where $\vert\Psi_{x,z}\rangle:=(\1\otimes W_{x,z})\vert\Psi^\star\rangle$.

Now we show that the GBM satisfies the two symmetry conditions required by
Theorem~\ref{thm:verification-efficiency}.  Let $G$ be the finite local Weyl group;
after suppressing irrelevant phases, its action is represented by
\begin{align}
V((x,z),(x',z') ):=W_{x,z} \otimes W_{x',z'}
\end{align}
for $((x,z),(x',z'))\in \bZ_d^2 \times \bZ_d^2$.
This action permutes the $d^2$ Bell states transitively.  The stabilizer of
$\vert\Psi^\star\rangle$ is generated, up to phases, by
$\{ W_{x,z} \otimes W_{x,-z} :(x,z) \in \bZ_d^2 \}$, under which the Bell basis
decomposes into one dimensional irreducible subspaces.  Hence the GBM is locally
transitive and irreducible.

By Theorem~\ref{thm:verification-efficiency}, the QMV problem constrained by locality
therefore reduces to a single symmetric pure state verification problem.  For the
reference outcome $(0,0)$, set $T=d^2\Omega_{0,0}(\cS)$.  
The optimal single-round worst-case passing probability can be written as
\begin{align}
 P^\star_{\cB,\mathrm{loc}}(\varepsilon)
 :=\min_{\cS\,\mathrm{local}}P(\cS,\varepsilon)
 =\min_T\max_{\sss\bra{\Psi^\star}\sigma\ket{\Psi^\star}\leq 1-\varepsilon}
  \tr [T \sigma],\label{NBF4}
\end{align}
where the minimization is over separable verification operators satisfying
$0\le T\le \1$ and $T\vert\Psi^\star\rangle=\vert\Psi^\star\rangle$.
Thus the QMV problem is reduced to separable verification of the maximally
entangled state $\vert\Psi^\star\rangle$.  The remaining task is to solve this
minimax problem and lift the optimal representative operator back to all
GBM outcomes.  The following corollary states the optimal local 
verification protocol and sample complexity; the proof and explicit MUB
construction are given in Appendix~\ref{appx:sec:GBM}.

\begin{corollary}\label{thm:GBM}
Assume that $d$ is prime.  The generalized Bell measurement $\cB$ defined in
Eq.~\eqref{eq:GBM} admits an optimal local verification protocol $\cS_{\cB}$ with
verification operators
\begin{align}\label{eq:GBM-verification-operator}
\Omega_{x,z}(\cS_{\cB})
&=\frac{1}{d^2}\Bigl[
|\Psi_{x,z}\rangle\langle \Psi_{x,z}| \notag\\
&\qquad
+\frac{1}{d+1}
\left(\1-|\Psi_{x,z}\rangle\langle \Psi_{x,z}|\right)
\Bigr].
\end{align}
Consequently,
\begin{align}
P(\cS_{\cB},\varepsilon)
=P^\star_{\cB,\mathrm{loc}}(\varepsilon)
=1-\frac{d}{d+1}\varepsilon,
\end{align}
and the exact number of tests needed to verify $\cB$ within infidelity
$\varepsilon$ and significance level $\delta$ is
\begin{align}\label{eq:GBM-copies}
  N_{\cB}
  = \left\lceil
  \frac{\ln\delta}{\ln P(\cS_{\cB},\varepsilon)}
  \right\rceil
\approx
  \left\lceil
  \frac{d+1}{d}\cdot\frac{1}{\varepsilon}\ln\frac{1}{\delta}
  \right\rceil.
\end{align}
\end{corollary}

Comparing Eq.~\eqref{eq:GBM-copies} with Eq.~\eqref{eq:optimal-copies},
we see that the locality constraint incurs only the constant factor penalty $(d+1)/d$
relative to the unconstrained optimal protocol.  This penalty is $3/2$ for the Bell
measurement on two qubits and approaches one as $d$ grows.  
For concreteness, the optimal QMV strategy for the two-qubit Bell measurement 
is listed in Table~\ref{tbl:bell-measurement-verification}.

\begin{table}[!hbtp]
\centering
\setlength\heavyrulewidth{0.3ex}
\setlength{\tabcolsep}{8pt}
\renewcommand{\arraystretch}{1.3}
\begin{tabular}{@{}cccc@{}}
    \toprule
      \textbf{Index $x$} &
      \textbf{Prob. $p_x$} & \textbf{Test State $\rho_x$} & \textbf{Accept Set $\cB(\rho_x)$} \\\midrule
      1 & $1/12$ & $\ket{00}$ & $\{00,01\}$ \\
      2 & $1/12$ & $\ket{11}$ & $\{00,01\}$ \\
      3 & $1/12$ & $\ket{01}$ & $\{10,11\}$ \\
      4 & $1/12$ & $\ket{10}$ & $\{10,11\}$ \\
      5 & $1/12$ & $\ket{++}$ & $\{00,10\}$ \\
      6 & $1/12$ & $\ket{--}$ & $\{00,10\}$ \\
      7 & $1/12$ & $\ket{+-}$ & $\{01,11\}$ \\
      8 & $1/12$ & $\ket{-+}$ & $\{01,11\}$ \\
      9 & $1/12$ & $\ket{\top\top}$    & $\{01,10\}$ \\
      10 & $1/12$ & $\ket{\perp\perp}$  & $\{01,10\}$ \\
      11 & $1/12$ & $\ket{\top\!\perp}$ & $\{00,11\}$ \\
      12 & $1/12$ & $\ket{\perp\!\top}$ & $\{00,11\}$ \\
    \bottomrule
\end{tabular}
\caption{\textbf{An optimal verification protocol for the two-qubit Bell measurement using local states.}
    Here, $\{\ket{0},\ket{1}\}$, $\{\ket{+},\ket{-}\}$, and $\{\ket{\top},\ket{\perp}\}$
    are the eigenstates of the single-qubit Pauli operators $Z$, $X$, and $Y$, respectively.}
\label{tbl:bell-measurement-verification}
\end{table}

\subsection{Single-parameter measurement on two qubits}\label{sec:single-parameter}

We next consider a single-parameter family of measurements on two qubits that interpolates
between product basis measurements and the Bell measurement.  Define the four
orthonormal states on two qubits
\begin{subequations}\label{eq:entangled-states}
\begin{align}
  \vert\Psi_{0,0}^{\gamma}\rangle &:= \sin\gamma\ket{00} + \cos\gamma\ket{11}, \\
  \vert\Psi_{0,1}^{\gamma}\rangle &:= \sin\gamma\ket{01} + \cos\gamma\ket{10}, \\
  \vert\Psi_{1,0}^{\gamma}\rangle &:= \cos\gamma\ket{01} - \sin\gamma\ket{10}, \\
  \vert\Psi_{1,1}^{\gamma}\rangle &:= \cos\gamma\ket{00} - \sin\gamma\ket{11},
\end{align}
\end{subequations}
where $\gamma\in[0,\pi/2]$.  The corresponding von Neumann measurement is
\begin{align}\label{eq:single-parameter}
  \cP_\gamma = \{\vert \Psi_{x,z}^\gamma\rangle\!\langle\Psi_{x,z}^\gamma\vert\}_{x,z\in\bZ_2}.
\end{align}
At $\gamma=\pi/4$, $\cP_\gamma$ is the Bell measurement on two qubits; at
$\gamma=0$ and $\gamma=\pi/2$, it reduces to a product basis measurement up to
relabeling.

We first show that the family $\cP_\gamma$ is locally transitive. 
Up to irrelevant phases, the local unitaries generated by
\begin{align}
  \1\otimes X,\quad Z\otimes Z,\quad ZX\otimes\1
\end{align}
permute the four projectors in Eq.~\eqref{eq:single-parameter}.  This allows us
to use a symmetric local protocol.  Appendix~\ref{appx:sec:single-parameter}
constructs such a protocol $\cS_\gamma$ explicitly.  The corollary below gives
the resulting verification operators, success probability, and sample
complexity.

\begin{corollary}\label{thm:single-parameter two-qubit measurement}
For every $\gamma\in[0,\pi/2]$, the single-parameter measurement on two qubits
$\cP_\gamma$ defined in Eq.~\eqref{eq:single-parameter} admits a local
verification protocol $\cS_\gamma$ with verification operators
\begin{align}\label{eq:single-parameter-verification-operator}
\Omega_{x,z}(\cS_\gamma)
=\frac{1}{4}\left[
\proj{\Psi_{x,z}^{\gamma}}
+\frac{1}{3}\left(\1-\proj{\Psi_{x,z}^{\gamma}}\right)
\right].
\end{align}
Consequently,
\begin{align}
P(\cS_\gamma,\varepsilon)=1-\frac{2}{3}\varepsilon,
\end{align}
and the exact number of tests required by this protocol is
\begin{align}
  N_{\cP_\gamma}
  =
  \left\lceil
  \frac{\ln\delta}{\ln P(\cS_\gamma,\varepsilon)}
  \right\rceil
  \approx
  \left\lceil
  \frac{3}{2}\cdot\frac{1}{\varepsilon}\ln\frac{1}{\delta}
  \right\rceil.
\end{align}
\end{corollary}

The protocol has the same asymptotic sample complexity as the Bell measurement
protocol in the previous subsection.  At the Bell point $\gamma=\pi/4$, it
reduces to Table~\ref{tbl:bell-measurement-verification}.  At the two product
endpoints, one can instead verify the product basis measurement with the
unconstrained optimal scaling in Eq.~\eqref{eq:optimal-copies}.  By the local
equivalence argument in Appendix~\ref{appx:sec:properties}, the same efficiency
applies to any measurement locally equivalent to $\cP_\gamma$.

\subsection{Elegant joint measurement on two qubits}\label{sec:elegant-joint-measurement}

The elegant joint measurement (EJM) on two qubits is a partially entangled
measurement with tetrahedral one qubit marginals.  It was introduced as a
joint measurement different from the Bell measurement~\cite{gisin2019entanglement}
and later used in bilocal network nonlocality tests
~\cite{tavakoli2021bilocal,PhysRevLett.129.030502}.  This makes the EJM a
representative target on two qubits beyond both product basis measurements and
maximally entangled Bell measurements.

In its general form, the EJM on two qubits is itself a single-parameter family.  Let
$\kappa\in[0,\pi/2]$ and $r_\pm^\kappa:=(1\pm e^{i\kappa})/\sqrt{2}$.
Define four orthonormal states~\cite{PhysRevLett.129.030502}:
\begin{subequations}\label{eq:EJM-states}
\begin{align}
    \vert\Phi_0^\kappa\rangle
&= \frac{1}{2}\left(e^{ - \frac{i\pi}{4}}\ket{00}
                    - r_+^\kappa\ket{01}
                    - r_-^\kappa\ket{10}
                    + e^{ - \frac{3i\pi}{4}}\ket{11}\right), \\
    \vert\Phi_1^\kappa\rangle
&= \frac{1}{2}\left(e^{\frac{i\pi}{4}}\ket{00}
                    + r_-^\kappa\ket{01}
                    + r_+^\kappa\ket{10}
                    + e^{\frac{3i\pi}{4}}\ket{11}\right), \\
    \vert\Phi_2^\kappa\rangle
&= \frac{1}{2}\left(e^{ - \frac{3i\pi}{4}}\ket{00}
                    + r_-^\kappa\ket{01}
                    + r_+^\kappa\ket{10}
                    + e^{- \frac{i\pi}{4}}\ket{11}\right), \\
    \vert\Phi_3^\kappa\rangle
&= \frac{1}{2}\left(e^{\frac{3i\pi}{4}}\ket{00}
                    - r_+^\kappa\ket{01}
                    - r_-^\kappa\ket{10}
                    + e^{\frac{i\pi}{4}}\ket{11}\right).
\end{align}
\end{subequations}
The single-parameter two-qubit EJM is defined as
\begin{align}\label{eq:EJM}
  \cJ_\kappa = \{\vert \Phi_j^\kappa\rangle\!\langle\Phi_j^\kappa\vert:
  j=0,1,2,3\}.
\end{align}
At $\kappa=0$ this recovers the original EJM in
Ref.~\cite{gisin2019entanglement}, while $\kappa=\pi/2$ gives a maximally
entangled endpoint.

For $0\le\kappa<\pi/2$, $\cJ_\kappa$ is essentially different from
the single-parameter family $\cP_\gamma$ considered in
Section~\ref{sec:single-parameter}.  Indeed, every state in $\cP_\gamma$ has
the same Schmidt basis up to local relabeling, whereas the states in
$\cJ_\kappa$ have tetrahedrally arranged reduced states.  Hence the verification
protocol constructed for $\cP_\gamma$ cannot be imported by local covariance,
and a separate local protocol is required.

The local Pauli operators $X\otimes X$ and $Z\otimes Z$ act on the four
EJM projectors as the permutations $(01)(23)$ and $(03)(12)$,
respectively, up to irrelevant phases of the underlying state vectors.
Hence the group generated by them acts transitively on the four outcomes
for every $\kappa\in[0,\pi/2]$, including the maximally entangled endpoint
$\kappa=\pi/2$.  Thus $\cJ_\kappa$ admits a symmetric local verification
protocol. Appendix~\ref{appx:sec:elegant-joint-measurement}
constructs this protocol explicitly; the corollary below gives 
its verification operators, success probability, and sample complexity.

\begin{corollary}\label{thm:elegant-joint-measurement}
For every $\kappa\in[0,\pi/2]$, the elegant joint measurement $\cJ_\kappa$
defined in Eq.~\eqref{eq:EJM} admits a local verification protocol
$\cS_\kappa$ with verification operators, for $j=0,1,2,3$,
\begin{align}\label{eq:EJM-verification-operator}
\Omega_j(\cS_\kappa)
=\frac{1}{4}\left[
\proj{\Phi_j^\kappa}
+\beta_\kappa\left(\1-\proj{\Phi_j^\kappa}\right)
\right],
\end{align}
where 
\begin{align}\label{eq:EJM-beta}
  \beta_\kappa=\frac{4-3\sin\kappa}{3(2-\sin\kappa)}.
\end{align}
Consequently,
\begin{align}
P(\cS_\kappa,\varepsilon)
=1-(1-\beta_\kappa)\varepsilon
=1-\frac{2}{3(2-\sin\kappa)}\varepsilon,
\end{align}
and the exact number of tests required by this protocol is
\begin{align}\label{eq:EJM-copies}
  N_{\cJ_\kappa}
  =
  \left\lceil
  \frac{\ln\delta}{\ln P(\cS_\kappa,\varepsilon)}
  \right\rceil
  \approx
  \left\lceil
  \frac{3(2-\sin\kappa)}{2}\cdot
  \frac{1}{\varepsilon}\ln\frac{1}{\delta}
  \right\rceil.
\end{align}
\end{corollary}

Compared with the homogeneous protocol in Corollary~\ref{thm:single-parameter two-qubit measurement}, 
the EJM protocol has a homogeneous parameter that depends on $\kappa$.  It equals $2/3$ for the original
EJM at $\kappa=0$ and decreases to $1/3$ at the maximally entangled endpoint
$\kappa=\pi/2$.  Thus the original EJM requires twice as many tests as the
Bell measurement scaling, while the endpoint recovers the same asymptotic sample
complexity as the Bell measurement.

\subsection{Stabilizer state induced measurements}\label{sec:stabilizer-state-induced-measurements}

We now consider measurements induced by stabilizer bases on $n$ qudits, assuming
that the local dimension $d$ is prime.  Let $\bF_d$ be the finite field with $d$
elements, and define the Weyl representation on $n$ qudits
\begin{align}
W^n(\vec{x} ):=
W(x_1) \otimes \cdots \otimes W(x_n)
\end{align}
for $\vec{x}=(x_1,\ldots,x_n)\in(\bF_d^2)^n\simeq\bF_d^{2n}$, where
$x_i=(a_i,b_i)$ and $W(x_i)\equiv W_{a_i,b_i}$ is a single qudit Weyl operator;
see Appendix~\ref{appx:sec:GBM}.

Let $N \subset \bF_d^{2n}$ be a commutative subgroup with dimension $n$.
For this subgroup, let $\ket{S}$ be the stabilizer state satisfying
\begin{align}
W^n(\vec{x})\ket{S}=\ket{S},\qquad \forall\vec{x}\in N.
\end{align}
The orbit
$\{\ket{S,\vec{a}}:=W^n(\vec{a})\ket{S}:\vec{a}\in\bF_d^{2n}/N\}$
forms an orthonormal basis, and the induced measurement is
\begin{align}\label{eq:stabilizer state induced measurements}
  \cM_S := \{ \proj{S,\vec{a}}: \vec{a} \in \bF_d^{2n}/N\}.
\end{align}

The Weyl translations $W^n(\vec{x})$ are local and act 
transitively on the outcomes of $\cM_S$, because 
$W^n(\vec{x})\ket{S,\vec a}$
is proportional to $\ket{S,\vec{x}+\vec a}$, with the label understood
modulo $N$.  The stabilizer of the reference outcome is precisely $N$,
up to irrelevant phases.  For $\vec n\in N$, the basis state
$\ket{S,\vec a}=W^n(\vec a)\ket S$ is an eigenvector of the stabilizer
operator $W^n(\vec n)$ with character
\begin{align}
  \chi_{\vec a}(\vec n)
  =
  \omega^{[\vec n,\vec a]},
\end{align}
where $[\cdot,\cdot]$ denotes the standard symplectic form over
$\bF_d$ and $\omega=\exp(2\pi i/d)$.

These characters are distinct on different cosets of $N$.  Indeed, if
$\chi_{\vec a}=\chi_{\vec b}$ on $N$, then
\begin{align}
    [\vec n,\vec a-\vec b]=0  \qquad \forall\,\vec n\in N .
\end{align}
Thus $\vec a-\vec b\in N^\perp$.  Since $N$ is maximal isotropic,
$N^\perp=N$, and hence $\vec a$ and $\vec b$ represent the same coset in
$\bF_d^{2n}/N$.  Therefore different cosets give inequivalent
one-dimensional characters of the stabilizer.  The reference line has no
equivalent copy in its orthogonal complement, so the irreducibility
condition in Definition~\ref{def:irreducibility} is satisfied.  
Therefore Theorem~\ref{thm:verification-efficiency} applies.
Appendix~\ref{appx:sec:stabilizer-state}
constructs the local protocol; the corollary below 
summarizes its verification operators, success probability, and sample complexity.

\begin{corollary}\label{thm:stabilizer state induced measurements}
The stabilizer state induced measurement $\cM_S$ defined in
Eq.~\eqref{eq:stabilizer state induced measurements} admits the local
verification protocol $\cS_{\cM_S}$. 
For each outcome $\vec{a}\in\bF_d^{2n}/N$, its verification operator is
\begin{align}\label{eq:stabilizer-verification-operator}
\Omega_{\vec{a}}(\cS_{\cM_S})
&=\frac{1}{d^n}
\left[\proj{S,\vec{a}}+\beta_S\left(\1-\proj{S,\vec{a}}\right)\right],
\end{align}
where
\begin{align}
\beta_S=\frac{d^{n-1}-1}{d^n-1}.
\end{align}
Consequently,
\begin{align}
P(\cS_{\cM_S},\varepsilon)
&=1-(1-\beta_S)\varepsilon
=1-\frac{d^n-d^{n-1}}{d^n-1}\varepsilon,
\end{align}
and the exact number of tests required by this protocol is
\begin{align}
  N_{\cM_S}
  =
  \left\lceil
  \frac{\ln\delta}
       {\ln P(\cS_{\cM_S},\varepsilon)}
  \right\rceil
  \approx
  \left\lceil
  \frac{d^n - 1}{d^n - d^{n-1}}\cdot
  \frac{1}{\varepsilon}\ln\frac{1}{\delta}
  \right\rceil.
\end{align}
\end{corollary}

\section{Measurement fidelity estimation}\label{sec:measurement-fidelity-estimation}

A QMV protocol $\cS$ for the measurement device $\mathscr{D}$ can be
utilized to estimate the device's fidelity, as in the cases of quantum state and
process verification~\cite{Zhu2019Optimal,Jiang2020Towards,Zhang2022Efficient,wollmann2023fidelity}.
In particular, because the constructed verification operators are homogeneous,
we can estimate the measurement fidelity in a more direct manner.
By the definition of homogeneous operator in Eq.~\eqref{eq:homogeneous},
QMV can be related to measurement fidelity estimation via
\begin{align}
  \Fid(\cP,\cM)
&= \frac{1}{\vert\Theta\vert}\sum_{\theta\in\Theta}\bra{\psi_\theta}M_\theta\ket{\psi_\theta} \\
&= \frac{\sum_{\theta\in\Theta}\tr[\Omega_\theta M_\theta] - \beta}{1 - \beta},
\end{align}
where the average passing probability $\sum_{\theta\in\Theta}\tr[\Omega_\theta M_\theta]$ can be evaluated
from the experimental data. Specifically, we consecutively perform the
measurement verification protocol $\cS$ for $N$ rounds and record the number of passing instances $f$.
Then $\sum_{\theta\in\Theta}\tr[\Omega_\theta M_\theta]\approx f/N$, as guaranteed
by the law of large numbers. However, for finite $N$, the standard deviation is bounded by
\begin{align}\label{eq:standard-deviation-main}
  \Delta(\widehat{F}) \leq \frac{1}{2(1-\beta)\sqrt{N}},
\end{align}
meaning that the number of rounds required must satisfy
$N=\cO(1/\varepsilon_{\rm est}^2)$
in order to achieve estimation precision $\varepsilon_{\rm est}$,
which is quadratically worse than that of verification.
Details regarding measurement fidelity estimation based on QMV
can be found in Appendix~\ref{appx:measurement-fidelity-estimation}.

Nevertheless, using QMV as an estimation procedure,
one can directly obtain the value of measurement fidelity rather than a bound.
In addition, performing fidelity estimation only needs to know the \emph{frequency} of pass instances
rather than the number of \emph{successive} pass cases, which is much more experimentally robust.

\section{Conclusions and outlook}\label{sec:conclusions}

We developed a framework for verifying entangled measurements using only local
state preparations.  We formulated the QMV task in terms of the worst-case
passing probability of faulty POVMs and proved that, for locally transitive and
irreducible von Neumann measurements, symmetry reduced the QMV problem
constrained by locality to a QSV style optimization for one representative basis state.  We
then constructed homogeneous local protocols for generalized Bell measurements,
single-parameter measurements on two qubits, elegant joint measurements, and
stabilizer state induced measurements, and derived their verification operators,
success probabilities, and sample complexities.  We also showed how homogeneous
QMV protocols could be used to estimate measurement fidelity directly from the
frequency of passing events.

Two problems deserve further study.  First, it is important to extend the
present symmetry reduction method beyond locally transitive and irreducible von
Neumann measurements. Second, a fully realistic theory should incorporate
imperfect test state preparation, device behavior that is not independent and identically distributed, and finite sample
confidence bounds, so that QMV protocols can be deployed robustly on quantum platforms.

\section*{Acknowledgements}
MH was supported in part by
the General R\&D Projects of 1+1+1 CUHK-CUHK(SZ)-GDST Joint Collaboration Fund (Grant No. GRDP2025-022) and the Guangdong Provincial Quantum Science Strategic Initiative (Grant No. 
GDZX2505003).

%

\makeatletter
\newcommand{\appendixtitle}[1]{\gdef\@title{#1}}
\makeatother

\makeatletter%
\newcommand{\appendixmaketitle}{%
\begin{center}%
\vspace{0.4in}%
{\large \@title \par}%
\end{center}%
\par%
}%
\makeatother%

\makeatletter
\newcommand{\appendixtableofcontents}{%
\begingroup
\setcounter{tocdepth}{4}%
\@starttoc{atoc}%
\endgroup}
\makeatother

\newcommand{\appendixtocdivider}{%
\par\medskip
\noindent\hbox to \linewidth{%
\leaders\hbox{\rule[0.6ex]{1pt}{0.4pt}}\hfill
\hspace{0.8em}\textsc{Table of Contents}\hspace{0.8em}%
\leaders\hbox{\rule[0.6ex]{1pt}{0.4pt}}\hfill}%
\par}

\makeatletter
\newcounter{subsubsubsection}[subsubsection]
\renewcommand{\thesubsubsubsection}{\thesubsubsection.\arabic{subsubsubsection}}
\providecommand{\theHsubsubsubsection}{}
\renewcommand{\theHsubsubsubsection}{\theHsubsubsection.\arabic{subsubsubsection}}
\providecommand*{\toclevel@subsubsubsection}{4}
\newcommand{\subsubsubsection}[1]{%
\refstepcounter{subsubsubsection}%
\paragraph*{\thesubsubsubsection\space #1}}
\newcommand*\l@subsubsubsection[2]{\@dottedtocline{4}{4.4em}{3.0em}{#1}{#2}}
\makeatother

\newcommand{\appsection}[1]{%
\section{#1}%
\addcontentsline{atoc}{section}{\protect\numberline{\thesection}#1}}

\newcommand{\appsubsection}[1]{%
\subsection{#1}%
\addcontentsline{atoc}{subsection}{\protect\numberline{\thesubsection}#1}}

\newcommand{\appsubsubsection}[1]{%
\subsubsection{#1}%
\addcontentsline{atoc}{subsubsection}{\protect\numberline{\thesubsubsection}#1}}

\newcommand{\appsubsubsubsection}[1]{%
\subsubsubsection{#1}}

\setcounter{secnumdepth}{4}
\appendix
\widetext
\newpage

\appendixtitle{\bf
Supplemental Material for\\``\thetitle''}
\appendixmaketitle
\vspace{0.1in}

This Supplemental Material is organized as follows.
Appendix~\ref{appx:sec:properties} collects basic properties of the maximal
passing probability used in QMV.
Appendix~\ref{appx:sec:measurement-symmetry} proves the symmetry reduction
developed in Section~\ref{sec:symmetries-simplify-QMV}.
Appendices~\ref{appx:sec:GBM}--\ref{appx:sec:stabilizer-state} provide the
detailed protocols, symmetry analysis, and proofs for the four representative
classes of measurements studied in Section~\ref{sec:applications}: generalized
Bell measurements, single-parameter measurements on two qubits, elegant joint
measurements, and stabilizer state induced measurements.
Appendix~\ref{appx:measurement-fidelity-estimation} contains the details of
measurement fidelity estimation based on homogeneous QMV protocols.

\appendixtocdivider
{%
\appendixtableofcontents%
}%

\appsection{Properties of the maximal probability}\label{appx:sec:properties}

In this section, we collect several basic properties of the maximal
passing probability $P(\cS,\varepsilon)$ defined in
Eq.~\eqref{eq:worst-case failure probability} of the main text.
These properties are used when designing efficient and optimal
verification protocols.

\begin{proposition}\label{prop:pure states}
It suffices to consider pure test states to minimize
$P(\cS,\varepsilon)$.
\end{proposition}

\begin{proof}
Consider an arbitrary QMV protocol
\begin{align}
  \cS=\{p_x,(\rho_x,\cP(\rho_x))\}_{x\in\cX}
\end{align}
for the target measurement $\cP$.  For each test state $\rho_x$, choose
a spectral decomposition
\begin{align}
  \rho_x=\sum_y q_{y\vert x} v_{y\vert x},
  \qquad
  v_{y\vert x}:=\proj{v_{y\vert x}},
  \label{eq:spectral-decomposition-rhox}
\end{align}
where only eigenvectors with $q_{y\vert x}>0$ are included.  Zero
eigenvalue components are omitted, and the corresponding eigenvectors
may be completed to an orthonormal basis of the underlying Hilbert space
$\cH^{\otimes n}$.

We refine the original protocol by replacing each mixed test state with
the pure states appearing in its spectral decomposition:
\begin{align}
  \cS'
  :=
  \left\{
  p_x q_{y\vert x},
  \left(v_{y\vert x},\cP(v_{y\vert x})\right)
  \right\}_{x,y}.
  \label{eq:pure-refined-protocol}
\end{align}
This protocol consists only of pure test states.  It also accepts the
ideal measurement $\cP$ with certainty, because the accept set is again
defined as the support of the ideal outcome distribution.

It remains to compare the worst-case passing probabilities.  Take any
POVM $\cM=\{M_\theta\}_{\theta\in\Theta}$ satisfying
\begin{align}
  \Fid(\cP,\cM)\leq 1-\varepsilon .
\end{align}
For every pair $(x,y)$ appearing in the above decomposition,
\begin{align}
  \cP(v_{y\vert x})\subseteq \cP(\rho_x).
\end{align}
Indeed, if $\theta\in\cP(v_{y\vert x})$, then
\begin{align}
  \bra{\psi_\theta}v_{y\vert x}\ket{\psi_\theta}>0.
\end{align}
Since $q_{y\vert x}>0$, Eq.~\eqref{eq:spectral-decomposition-rhox}
implies
\begin{align}
  \bra{\psi_\theta}\rho_x\ket{\psi_\theta}
  \geq
  q_{y\vert x}
  \bra{\psi_\theta}v_{y\vert x}\ket{\psi_\theta}
  >0.
\end{align}
Thus $\theta\in\cP(\rho_x)$.

For this fixed POVM $\cM$, we therefore have
\begin{align}
p_{\rm pass}(\cS',\cM)
&=
\sum_{x,y}p_xq_{y\vert x}
\sum_{\theta\in\cP(v_{y\vert x})}
\tr[v_{y\vert x}M_\theta] \\
&\leq
\sum_{x,y}p_xq_{y\vert x}
\sum_{\theta\in\cP(\rho_x)}
\tr[v_{y\vert x}M_\theta] \\
&=
\sum_xp_x
\sum_{\theta\in\cP(\rho_x)}
\tr\!\left[
\left(\sum_yq_{y\vert x}v_{y\vert x}\right)M_\theta
\right] \\
&=
\sum_xp_x
\sum_{\theta\in\cP(\rho_x)}
\tr[\rho_xM_\theta]
=
p_{\rm pass}(\cS,\cM).
\end{align}
Maximizing over all POVMs satisfying
$\Fid(\cP,\cM)\leq1-\varepsilon$ gives
\begin{align}
    P(\cS',\varepsilon)\leq P(\cS,\varepsilon).
\end{align}
Hence any protocol with mixed test states can be refined to a pure-state
protocol without increasing the maximal passing probability.

When the test states are required to be local product states, the same
argument is compatible with the locality constraint for product mixed
states: it suffices to take local spectral decompositions, thereby obtaining
pure local product test states.
\end{proof}

\begin{proposition}\label{prop:nonincreasing}
$P(\cS,\varepsilon)$ is nonincreasing in $\varepsilon$.  Namely, for
$\varepsilon,\varepsilon'\in(0,1)$ with
$\varepsilon\leq\varepsilon'$,
\begin{align}
  P(\cS,\varepsilon)\geq P(\cS,\varepsilon').
\end{align}
\end{proposition}

\begin{proof}
Increasing $\varepsilon$ makes the bad-measurement condition more
restrictive.  More precisely, since $\varepsilon\leq\varepsilon'$, the
feasible set
\begin{align}
    \{\cM:\Fid(\cP,\cM)\leq1-\varepsilon'\}
\end{align}
is contained in
\begin{align}
    \{\cM:\Fid(\cP,\cM)\leq1-\varepsilon\}.
\end{align}
The same objective function $p_{\rm pass}(\cS,\cM)$ is maximized over
these two feasible sets.  Maximization over the larger feasible set
cannot give a smaller value, and hence
\begin{align}
  P(\cS,\varepsilon) \geq P(\cS,\varepsilon').
\end{align}
\end{proof}

\begin{proposition}\label{prop:concavity}
$P(\cS,\varepsilon)$ is concave in $\varepsilon$.  Namely, for
$\varepsilon_1,\varepsilon_2\in(0,1)$ and $\lambda\in[0,1]$,
\begin{align}
  P(\cS,\lambda\varepsilon_1+(1-\lambda)\varepsilon_2)
  \geq
  \lambda P(\cS,\varepsilon_1)
  +(1-\lambda)P(\cS,\varepsilon_2).
\end{align}
\end{proposition}

\begin{proof}
The maximum in the definition of $P(\cS,\varepsilon)$ is attained:
the POVM set is compact in finite dimension, and both the objective
function and the fidelity constraint are continuous.

Choose maximizing POVMs
\begin{align}
    \cM^{(i)}=\{M_\theta^{(i)}\}_{\theta\in\Theta},
  \qquad i=1,2,
\end{align}
for $P(\cS,\varepsilon_1)$ and $P(\cS,\varepsilon_2)$, respectively.
Thus
\begin{align}
  \Fid(\cP,\cM^{(i)})\leq1-\varepsilon_i,
  \qquad
  P(\cS,\varepsilon_i)=p_{\rm pass}(\cS,\cM^{(i)}).
\end{align}
Their convex combination
\begin{align}
  \cM^\ast :=
  \{\lambda M_\theta^{(1)}+(1-\lambda)M_\theta^{(2)}\}_{\theta\in\Theta}
\end{align}
is again a POVM.  By linearity of the fidelity,
\begin{align}
  \Fid(\cP,\cM^\ast)
  &=
  \lambda\Fid(\cP,\cM^{(1)})
  +(1-\lambda)\Fid(\cP,\cM^{(2)}) \\
  &\leq
  \lambda(1-\varepsilon_1)
  +(1-\lambda)(1-\varepsilon_2) \\
  &=
  1-\{\lambda\varepsilon_1+(1-\lambda)\varepsilon_2\}.
\end{align}
Therefore $\cM^\ast$ is feasible for
$P(\cS,\lambda\varepsilon_1+(1-\lambda)\varepsilon_2)$.  
Consequently,
\begin{align}
P(\cS,\lambda\varepsilon_1+(1-\lambda)\varepsilon_2)
&\geq p_{\rm pass}(\cS,\cM^\ast) \\
&=
\lambda p_{\rm pass}(\cS,\cM^{(1)})
+(1-\lambda)p_{\rm pass}(\cS,\cM^{(2)}) \\
&= \lambda P(\cS,\varepsilon_1) + (1-\lambda)P(\cS,\varepsilon_2),
\end{align}
where the middle equality uses the linearity of
$p_{\rm pass}(\cS,\cM)$ in the POVM elements.
\end{proof}

\vspace{0.1in}
\textbf{Convex mixing of verification protocols.}
Consider two verification protocols for the same target measurement $\cP$,
\begin{align}
  \cS_1=\{p_x,(\rho_x,\cP(\rho_x))\}_x,
  \qquad
  \cS_2=\{q_y,(\sigma_y,\cP(\sigma_y))\}_y .
\end{align}
For $\lambda\in[0,1]$, their mixture is the protocol
\begin{align}
  \lambda\cS_1+(1-\lambda)\cS_2
  :=
  \left\{
  \lambda p_x,(\rho_x,\cP(\rho_x));
  \ (1-\lambda)q_y,(\sigma_y,\cP(\sigma_y))
  \right\}_{x,y}.
\end{align}

\begin{proposition}\label{prop:convexity}
$P(\cS,\varepsilon)$ is convex in the protocol $\cS$.  
More precisely, for $\cS=\lambda\cS_1+(1-\lambda)\cS_2$, one has
\begin{align}
  P(\cS,\varepsilon)
  \leq
  \lambda P(\cS_1,\varepsilon)
  +(1-\lambda)P(\cS_2,\varepsilon).
\end{align}
\end{proposition}

\begin{proof}
For every POVM $\cM$, the passing probability is affine in the protocol:
\begin{align}
  p_{\rm pass}(\lambda\cS_1+(1-\lambda)\cS_2,\cM)
  =
  \lambda p_{\rm pass}(\cS_1,\cM)
  +(1-\lambda)p_{\rm pass}(\cS_2,\cM).
\end{align}
Hence
\begin{align}
P(\cS,\varepsilon)
&=
\max_{\cM:\Fid(\cP,\cM)\leq1-\varepsilon}
\left[
\lambda p_{\rm pass}(\cS_1,\cM)
+(1-\lambda)p_{\rm pass}(\cS_2,\cM)
\right] \\
&\leq
\lambda
\max_{\cM:\Fid(\cP,\cM)\leq1-\varepsilon}
p_{\rm pass}(\cS_1,\cM)
+
(1-\lambda)
\max_{\cM:\Fid(\cP,\cM)\leq1-\varepsilon}
p_{\rm pass}(\cS_2,\cM) \\
&=
\lambda P(\cS_1,\varepsilon)
+(1-\lambda)P(\cS_2,\varepsilon).
\end{align}
\end{proof}

\vspace{0.1in}
\textbf{Local equivalence of target measurements.} Suppose that
\begin{align}
  \cP=\{\proj{\psi_\theta}\}_{\theta\in\Theta},
  \qquad
  \cQ=\{\proj{\phi_\theta}\}_{\theta\in\Theta}
\end{align}
are two projective measurements on $\cH^{\otimes n}$.  They are called
locally equivalent if local unitaries $U_1,\ldots,U_n$ exist such that,
with $U:=U_1\ox\cdots\ox U_n$,
\begin{align}
  \ket{\phi_\theta}=U\ket{\psi_\theta},
  \qquad
  \forall \theta\in\Theta.
\end{align}
The next proposition states that local equivalence preserves the maximal
passing probability.  To avoid ambiguity, $P_{\cP}$ and $P_{\cQ}$ denote
the maximal passing probabilities evaluated with respect to the target
measurements $\cP$ and $\cQ$, respectively.

\begin{proposition}\label{prop:locally equivalent}
For any protocol $\cS$ verifying $\cP$, there is a corresponding protocol
$\cS'$ verifying $\cQ$ such that
\begin{align}
  P_{\cQ}(\cS',\varepsilon)=P_{\cP}(\cS,\varepsilon).
\end{align}
\end{proposition}

\begin{proof}
Start from a protocol $\cS=\{p_x,(\rho_x,\cP(\rho_x))\}_x$ for $\cP$. 
The corresponding protocol for $\cQ$ is obtained by applying
the same local unitary to every test state:
\begin{align}
  \cS'
  :=
  \left\{
  p_x,
  \left(
  U\rho_xU^\dagger,\,
  \cQ(U\rho_xU^\dagger)
  \right)
  \right\}_x .
\end{align}
This protocol accepts the ideal measurement $\cQ$ with certainty.

The accept sets are transported by the same unitary.  Indeed, for every
$x$ and $\theta$,
\begin{align}
  \bra{\psi_\theta}\rho_x\ket{\psi_\theta}
  =
  \bra{\phi_\theta}U\rho_xU^\dagger\ket{\phi_\theta},
\end{align}
and therefore
\begin{align}
  \cP(\rho_x)=\cQ(U\rho_xU^\dagger).
  \label{eq:accept-sets-local-equivalence}
\end{align}

Now consider an arbitrary POVM
$\cM=\{M_\theta\}_{\theta\in\Theta}$ in the $\cQ$-verification problem.
After conjugation by $U$, define
\begin{align}
  M_\theta' := U^\dagger M_\theta U,
  \qquad
  \cM':=\{M_\theta'\}_{\theta\in\Theta}.
\end{align}
The family $\cM'$ is a valid POVM.  Moreover,
\begin{align}
\Fid(\cQ,\cM)
=
\frac{1}{|\Theta|}
\sum_{\theta\in\Theta}
\bra{\phi_\theta}M_\theta\ket{\phi_\theta}
=
\frac{1}{|\Theta|}
\sum_{\theta\in\Theta}
\bra{\psi_\theta}U^\dagger M_\theta U\ket{\psi_\theta}
=
\Fid(\cP,\cM').
\end{align}
The passing probabilities are also preserved:
\begin{align}
p_{\rm pass}(\cS',\cM)
=
\sum_xp_x
\sum_{\theta\in\cQ(U\rho_xU^\dagger)}
\tr[U\rho_xU^\dagger M_\theta]
=
\sum_xp_x
\sum_{\theta\in\cP(\rho_x)}
\tr[\rho_x U^\dagger M_\theta U]
=
p_{\rm pass}(\cS,\cM').
\end{align}
Thus conjugation by $U$ gives a bijection between the feasible POVMs for
$P_{\cQ}(\cS',\varepsilon)$ and those for
$P_{\cP}(\cS,\varepsilon)$, while preserving the objective value.
Consequently,
\begin{align}
  P_{\cQ}(\cS',\varepsilon)
  =
  P_{\cP}(\cS,\varepsilon).
\end{align}
\end{proof}

\appsection{Symmetry reduction for QMV with locality constraint}\label{appx:sec:measurement-symmetry}

This appendix proves the symmetry reduction statements claimed in
Section~\ref{sec:symmetries-simplify-QMV}. Throughout this section, $\cP$ is
locally transitive with respect to the finite group $G$ and the local unitary
representation $V$.

\appsubsection{Proof of Proposition~\ref{prop:symmetric-protocols}}
\label{appx:B1}
For any local protocol
$\cS=\{p_x,(\rho_x,\cP(\rho_x))\}_{x\in\cX}$, its symmetrization
$\bar{\cS}$ is local because each $V(g)$ is a product unitary.
For any POVM $\cM=\{M_\theta\}_{\theta\in\Theta}$, define the averaged POVM
\begin{align}
\bar{M}_\theta:=
\frac{1}{|G|}\sum_{g\in G}V^\dagger(g)M_{g\theta}V(g).
\end{align}
Then $\bar{\cM}:=\{\bar{M}_\theta\}_{\theta\in\Theta}$ is a valid POVM and
\begin{align}
\Fid(\cP,\bar{\cM})
=
\frac{1}{|\Theta|}
\sum_{\theta\in\Theta}
\bra{\psi_\theta}\bar{M}_\theta\ket{\psi_\theta}
=
\frac{1}{|G||\Theta|}
\sum_{g\in G}\sum_{\theta\in\Theta}
\bra{\psi_{g\theta}}M_{g\theta}\ket{\psi_{g\theta}}
=
\Fid(\cP,\cM).
\end{align}
Moreover, using $g\cP(\rho_x)=\cP(V(g)\rho_xV^\dagger(g))$, we have
\begin{align}
p_{\rm pass}(\bar{\cS},\cM)
=
\sum_{x,g}\frac{p_x}{|G|}
\sum_{\eta\in g\cP(\rho_x)}
\tr\!\left[V(g)\rho_xV^\dagger(g)M_\eta\right]
=
\sum_xp_x
\sum_{\theta\in\cP(\rho_x)}
\tr\!\left[\rho_x\bar M_\theta\right]
=
p_{\rm pass}(\cS,\bar{\cM}).
\end{align}

Now take any POVM $\cM$ satisfying $\Fid(\cP,\cM)\leq1-\varepsilon$.
Since $\Fid(\cP,\bar{\cM})=\Fid(\cP,\cM)$, the POVM $\bar{\cM}$ also satisfies
\begin{align}
    \Fid(\cP,\bar{\cM})\leq1-\varepsilon.
\end{align}
Therefore, by the definition of $P(\cS,\varepsilon)$,
\begin{align}
  p_{\rm pass}(\bar{\cS},\cM)
  =
  p_{\rm pass}(\cS,\bar{\cM})
  \leq
  P(\cS,\varepsilon).
\end{align}
Taking the maximum over all POVMs $\cM$ satisfying
$\Fid(\cP,\cM)\leq1-\varepsilon$ gives
\begin{align}
  P(\bar{\cS},\varepsilon)\leq P(\cS,\varepsilon).
\end{align}

\appsubsection{Proof of Theorem~\ref{thm:verification-efficiency}}
\label{appx:sec:proof-of-theorem-verification-efficiency}

This appendix first proves a lifting lemma for stabilizer-twirled states
and then uses it to reduce the QMV analysis to QSV for one representative
outcome.

\appsubsubsection{Preparation of a key lemma}

To prove Theorem~\ref{thm:verification-efficiency} from the main text, 
the following key lemma is needed to lift a stabilizer-twirled state to a POVM.
\begin{lemma}
\label{lem:povm-lifting}
Let $\cP=\{\proj{\psi_\theta}\}_{\theta\in\Theta}$ be locally transitive
and irreducible with respect to $G$ and $V$.  Fix
$\theta_0\in\Theta$, and let
\begin{align}
    H_{\theta_0}:=\{g\in G:g\theta_0=\theta_0\}.
\end{align}
For any density operator $\sigma$, define its stabilizer twirl by
\begin{align}
  \sigma_H
  :=
  \frac{1}{|H_{\theta_0}|}
  \sum_{h\in H_{\theta_0}}
  V^\dagger(h)\sigma V(h).
\end{align}
Choose one representative $g_\theta\in G$ for each
$\theta\in\Theta$ such that $g_\theta\theta_0=\theta$, and set
\begin{align}
  M_\theta
  :=
  V(g_\theta)\sigma_HV^\dagger(g_\theta),
  \qquad \theta\in\Theta .
\end{align}
Then $\{M_\theta\}_{\theta\in\Theta}$ is a POVM; that is,
\begin{align}
  M_\theta\geq0,\qquad\sum_{\theta\in\Theta}M_\theta=\1.
\end{align}
\end{lemma}
\begin{proof}
Positivity of each $M_\theta$ is immediate.  It remains to prove
\begin{align}
    \sum_{\theta\in\Theta}M_\theta=\1.
\end{align}
Let $A:=\sum_{\theta\in\Theta}M_\theta$.
Since $\sigma_H$ is invariant under the stabilizer $H_{\theta_0}$, the
definition of $M_\theta$ is independent of the choice of the
representative $g_\theta$.  Therefore the covariance relation
\begin{align}
V(k)AV^\dagger(k)=A, \qquad \forall k\in G,
\end{align}
holds.
Thus $A$ commutes with the full group representation.  In particular,
$A\ket{\psi_{\theta_0}}$ transforms under $H_{\theta_0}$ in the same
representation as the line spanned by $\ket{\psi_{\theta_0}}$.  By the
irreducibility assumption, no equivalent irreducible subspace occurs in
the orthogonal complement of this line.  Hence
\begin{align}
A\ket{\psi_{\theta_0}}=c\ket{\psi_{\theta_0}}
\end{align}
for some scalar $c$.

Since $A$ commutes with all $V(g)$ and $G$ acts transitively on
$\Theta$, the same scalar $c$ is obtained on every basis vector
$\ket{\psi_\theta}$.  Therefore 
\begin{align}
A=c\1.
\end{align}
Taking traces gives
\begin{align}
  c|\Theta|
  =
  \tr A
  =
  \sum_{\theta\in\Theta}\tr M_\theta
  =
  |\Theta|,
\end{align}
because each $M_\theta$ is unitarily equivalent to the density operator
$\sigma_H$ and hence has trace one.  Thus $c=1$, and $A=\1$.
Hence $\{M_\theta\}_{\theta\in\Theta}$ is a POVM.
\end{proof}

\appsubsubsection{Reduction to quantum state verification}

Fix an arbitrary outcome $\theta_0\in\Theta$.
The proof of Theorem~\ref{thm:verification-efficiency} proceeds in five steps:  
\begin{itemize}
  \item \textbf{Step 1} records the covariance and normalization properties forced by symmetry.  
  \item \textbf{Steps 2 and 3} prove matching upper and lower bounds for a fixed representative outcome $\theta_0$: 
   an arbitrary faulty POVM is averaged into a single density operator, and, conversely, 
   the lifting lemma converts any admissible density operator back
   into a POVM with the same objective value.  
   \item \textbf{Step 4} uses the same lifting  argument to 
  prove that $T_{\theta_0}(\cS)$ is a valid verification operator.
  \item \textbf{Step 5} removes the dependence on $\theta_0$ by transitivity and
applies the standard QSV eigenvalue formula.
\end{itemize}

\vspace*{0.1in}
\textbf{Step 1: Covariance and normalization.}
We first establish the covariance relations
\begin{align}
  V(g)\Omega_\theta(\cS)V^\dagger(g)
  =
  \Omega_{g\theta}(\cS),
  \qquad g\in G,\ \theta\in\Theta,
  \label{eq:symmetric-operator-covariance}
\end{align}
and
\begin{align}
  V(g)T_\theta(\cS)V^\dagger(g)
  =
  T_{g\theta}(\cS),
  \qquad g\in G,\ \theta\in\Theta,
  \label{eq:T-covariance}
\end{align}
together with the normalization
\begin{align}
  \bra{\psi_\theta}T_\theta(\cS)\ket{\psi_\theta}=1,
  \qquad \forall\theta\in\Theta .
  \label{eq:T-ideal-overlap-one}
\end{align}

Indeed, Eq.~\eqref{eq:symmetric-operator-covariance} follows directly
from the symmetry of $\cS$.  Since
\begin{align}
    T_\theta(\cS):=|\Theta|\Omega_\theta(\cS),
\end{align}
Eq.~\eqref{eq:T-covariance} follows immediately.  Moreover,
\begin{align}
\sum_{\theta\in\Theta}
\bra{\psi_\theta}\Omega_\theta(\cS)\ket{\psi_\theta}
=
\sum_x p_x
\sum_{\theta\in\cP(\rho_x)}
\bra{\psi_\theta}\rho_x\ket{\psi_\theta}
=
\sum_x p_x\tr[\rho_x]
=
1.
\end{align}
By Eq.~\eqref{eq:symmetric-operator-covariance} and transitivity, all
terms on the left-hand side are equal.  Therefore, for every
$\theta\in\Theta$,
\begin{align}
    |\Theta|\bra{\psi_\theta}\Omega_\theta(\cS)\ket{\psi_\theta}=1,
\end{align}
which is Eq.~\eqref{eq:T-ideal-overlap-one}.

\vspace*{0.1in}
\textbf{Step 2: From a POVM to a density operator.}
This step proves the upper bound
\begin{align}
  P(\cS,\varepsilon)
  \leq
  \max_{\substack{\sigma\geq0,\ \tr[\sigma]=1:\\
  \bra{\psi_{\theta_0}}\sigma\ket{\psi_{\theta_0}}\leq1-\varepsilon}}
  \tr[T_{\theta_0}(\cS)\sigma].
  \label{eq:upper-bound-qsv}
\end{align}

Let $\cM=\{M_\theta\}_{\theta\in\Theta}$ be an arbitrary POVM.  Define
\begin{align}
  \sigma_{\theta_0}(\cM)
  :=
  \frac{1}{|G|}
  \sum_{\theta\in\Theta}
  \sum_{\substack{g\in G:\\ g\theta_0=\theta}}
  V^\dagger(g)M_\theta V(g).
  \label{eq:sigma-from-povm}
\end{align}
This operator is positive.  Its trace is one because each
$\theta\in\Theta$ has $|H_{\theta_0}|$ preimages under the map
$g\mapsto g\theta_0$ and
\begin{align}
    |G|=|\Theta|\,|H_{\theta_0}|.
\end{align}
Thus $\sigma_{\theta_0}(\cM)$ is a density operator.

Using
\begin{align}
    V(g)\proj{\psi_{\theta_0}}V^\dagger(g) = \proj{\psi_{g\theta_0}},
\end{align}
we obtain
\begin{align}
  \bra{\psi_{\theta_0}}\sigma_{\theta_0}(\cM)
  \ket{\psi_{\theta_0}}
  =
  \Fid(\cP,\cM).
  \label{eq:fidelity-sigma}
\end{align}
Similarly, using Eq.~\eqref{eq:symmetric-operator-covariance}, we have
\begin{align}
  p_{\rm pass}(\cS,\cM)
  =
  \sum_{\theta\in\Theta}
  \tr[\Omega_\theta(\cS)M_\theta]
  =
  \tr[T_{\theta_0}(\cS)\sigma_{\theta_0}(\cM)].
  \label{eq:pass-to-state}
\end{align}
Therefore, if $\cM$ satisfies $\Fid(\cP,\cM)\leq1-\varepsilon$,
then $\sigma_{\theta_0}(\cM)$ satisfies
\begin{align}
    \bra{\psi_{\theta_0}}\sigma_{\theta_0}(\cM)
  \ket{\psi_{\theta_0}}
  \leq1-\varepsilon.
\end{align}
Taking the maximum over all such POVMs $\cM$ gives
Eq.~\eqref{eq:upper-bound-qsv}.

\vspace*{0.1in}
\textbf{Step 3: The reverse inequality.}
We next prove the reverse bound
\begin{align}
  P(\cS,\varepsilon)
  \geq
  \max_{\substack{\sigma\geq0,\ \tr[\sigma]=1:\\
  \bra{\psi_{\theta_0}}\sigma\ket{\psi_{\theta_0}}\leq1-\varepsilon}}
  \tr[T_{\theta_0}(\cS)\sigma].
  \label{eq:lower-bound-qsv}
\end{align}
Let $\sigma$ be any operator satisfying
\begin{align}
  \sigma\geq0,\qquad
  \tr[\sigma]=1,\qquad
  \bra{\psi_{\theta_0}}\sigma\ket{\psi_{\theta_0}}
  \leq1-\varepsilon.
\end{align}
Twirl it over the stabilizer $H_{\theta_0}$:
\begin{align}
  \sigma_H
  :=
  \frac{1}{|H_{\theta_0}|}
  \sum_{h\in H_{\theta_0}}
  V^\dagger(h)\sigma V(h).
  \label{eq:stabilizer-twirl}
\end{align}
Since $h\theta_0=\theta_0$ for $h\in H_{\theta_0}$, Eq.~\eqref{eq:T-covariance}
implies that $T_{\theta_0}(\cS)$ is invariant under the stabilizer.
The projector $\proj{\psi_{\theta_0}}$ is also invariant under the
stabilizer.  Hence
\begin{align}
  \tr[T_{\theta_0}(\cS)\sigma_H]
  &=
  \tr[T_{\theta_0}(\cS)\sigma],
  \label{eq:twirl-preserves-objective}\\
  \bra{\psi_{\theta_0}}\sigma_H\ket{\psi_{\theta_0}}
  &=
  \bra{\psi_{\theta_0}}\sigma\ket{\psi_{\theta_0}}.
  \label{eq:twirl-preserves-fidelity}
\end{align}

Choose one representative $g_\theta\in G$ for each $\theta\in\Theta$
such that $g_\theta\theta_0=\theta$.  By
Lemma~\ref{lem:povm-lifting}, the operators
\begin{align}
  M_\theta
  :=
  V(g_\theta)\sigma_HV^\dagger(g_\theta),
  \qquad \theta\in\Theta,
  \label{eq:lifted-povm-theorem}
\end{align}
form a POVM.

For this POVM, covariance gives
\begin{align}
  \Fid(\cP,\cM)
  &=
  \bra{\psi_{\theta_0}}\sigma_H\ket{\psi_{\theta_0}}
  \leq1-\varepsilon,
  \label{eq:lifted-povm-fidelity}\\
  p_{\rm pass}(\cS,\cM)
  &=
  \tr[T_{\theta_0}(\cS)\sigma_H]
  =
  \tr[T_{\theta_0}(\cS)\sigma].
  \label{eq:lifted-povm-pass}
\end{align}
Therefore
\begin{align}
  P(\cS,\varepsilon)
  \geq
  \tr[T_{\theta_0}(\cS)\sigma].
\end{align}
Since this holds for every $\sigma$ satisfying the constraints in
Eq.~\eqref{eq:lower-bound-qsv}, Eq.~\eqref{eq:lower-bound-qsv} follows.

Combining Eqs.~\eqref{eq:upper-bound-qsv} and~\eqref{eq:lower-bound-qsv}, we obtain
\begin{align}
  P(\cS,\varepsilon)
  =
  \max_{\substack{\sigma\geq0,\tr[\sigma]=1:\\
  \bra{\psi_{\theta_0}}\sigma\ket{\psi_{\theta_0}}\leq1-\varepsilon}}
  \tr[T_{\theta_0}(\cS)\sigma].
  \label{eq:qmv-qsv-fixed-theta}
\end{align}

\vspace*{0.1in}
\textbf{Step 4: Validity of $T_{\theta_0}(\cS)$ as a verification operator.}
We prove
\begin{align}
  0\leq T_{\theta_0}(\cS)\leq\1,
  \qquad
  T_{\theta_0}(\cS)\ket{\psi_{\theta_0}}
  =
  \ket{\psi_{\theta_0}}.
  \label{eq:T-theta0-valid}
\end{align}
The positivity $T_{\theta_0}(\cS)\geq0$ is immediate from the definition.
To prove $T_{\theta_0}(\cS)\leq\1$, let $\sigma$ be an arbitrary density
operator.  Apply the stabilizer twirl in Eq.~\eqref{eq:stabilizer-twirl}
and the lifting construction in Lemma~\ref{lem:povm-lifting} to this
$\sigma$.  As in Eq.~\eqref{eq:lifted-povm-pass}, the resulting POVM
$\cM$ satisfies
\begin{align}
    p_{\rm pass}(\cS,\cM) = \tr[T_{\theta_0}(\cS)\sigma].
\end{align}
Since $p_{\rm pass}(\cS,\cM)\leq1$ for every POVM $\cM$, we have
\begin{align}
    \tr[T_{\theta_0}(\cS)\sigma]\leq1
\end{align}
for every density operator $\sigma$.  Hence $T_{\theta_0}(\cS)\leq\1$.
Together with Eq.~\eqref{eq:T-ideal-overlap-one}, this implies
\begin{align}
    T_{\theta_0}(\cS)\ket{\psi_{\theta_0}} = \ket{\psi_{\theta_0}}.
\end{align}
Thus Eq.~\eqref{eq:T-theta0-valid} holds.

\vspace*{0.1in}
\textbf{Step 5: Independence of the representative outcome and the
eigenvalue formula.}
We finally show that Eq.~\eqref{eq:qmv-to-qsv} holds for every fixed
$\theta\in\Theta$ and that the spectrum of $T_\theta(\cS)$ is independent
of $\theta$.
Since the representative outcome $\theta_0$ was arbitrary,
Eq.~\eqref{eq:qmv-qsv-fixed-theta} holds with $\theta_0$ replaced by any
fixed $\theta\in\Theta$.  This proves Eq.~\eqref{eq:qmv-to-qsv}.  The same
argument also shows that Eq.~\eqref{eq:T-theta0-valid} holds with
$\theta_0$ replaced by any $\theta\in\Theta$, so each $T_\theta(\cS)$ is
a valid verification operator for $\ket{\psi_\theta}$.

Moreover, by local transitivity, for any $\theta,\eta\in\Theta$ there
exists $g\in G$ such that $\eta=g\theta$.  Eq.~\eqref{eq:T-covariance} then gives
\begin{align}
    T_\eta(\cS) = V(g)T_\theta(\cS)V^\dagger(g).
\end{align}
Thus all $T_\theta(\cS)$ are unitarily equivalent and have the same
spectrum.  In particular, $\lambda_2(T_\theta(\cS))$ is independent of
$\theta$.

Denote this common value by $\lambda_2(\cS)$.  Applying the standard
state-verification formula to the valid verification operator
$T_\theta(\cS)$ gives
\begin{align}
  P(\cS,\varepsilon)
  =
  1-\bigl[1-\lambda_2(\cS)\bigr]\varepsilon.
\end{align}
This completes the proof of Theorem~\ref{thm:verification-efficiency}.

\appsection{Generalized Bell measurement verification}\label{appx:sec:GBM}

\appsubsection{Introduction to generalized Bell measurements}

We first define the Weyl operators, which generalize the qubit Pauli operators.
For every positive integer $d$, define the set $\bZ_d:=\{0,\cdots, d-1\}$.
Whenever elements of $\bZ_d$ appear in arithmetic expressions, we assume that the operations are to be taken modulo $d$.
The \textit{discrete Weyl operators} acting on a Hilbert space $\cH$ of dimension $d$
are defined as follows.  Set
\begin{align}
    \zeta := \exp\left(\frac{2\pi i}{d}\right),
\end{align}
the root of unity of order $d$, and define the unitary operators
\begin{align}
    X := \sum_{c\in\bZ_d} \ketbra{c+1}{c},\quad
    Z := \sum_{c\in\bZ_d} \zeta^{c}\proj{c}.
\end{align}
When $d=2$, $X$ and $Z$ are the standard Pauli operators.
For each index pair $(x,z)\in\bZ_d\times\bZ_d$, the discrete Weyl operator $W_{x,z}$ in $\cH$ is defined as
\begin{align}\label{eq:Weyl-operators}
  W_{x,z} := X^x Z^z = \sum_{c\in\bZ_d} \zeta^{zc} \ketbra{c+x}{c}.
\end{align}
When $d=2$, we recover the complete set of Pauli operators.

From the canonical maximally entangled state of two qudits $\vert\Psi^\star\rangle$
defined in Eq.~\eqref{eq:entangled-qudit-mes} of the main text,
the generalized Bell states $\{\vert\Psi_{x,z}\rangle:x,z\in\bZ_d\}$
are defined via the Weyl operators by
\begin{align}
   \vert\Psi_{x,z}\rangle
:= (\1_R\otimes W_{x,z})\vert\Psi^\star\rangle
 = \frac{1}{\sqrt{d}}\sum_{i=0}^{d-1}\ket{i}\otimes W_{x,z}\ket{i}.
\end{align}
Based on the generalized Bell states, the \emph{generalized Bell measurement} is defined as
\begin{align}\label{eq:GBM-2}
  \cB = \{\vert\Psi_{x,z}\rangle\!\langle\Psi_{x,z}\vert: x,z\in \bZ_d\}.
\end{align}
Since all the Bell states are entangled,
the generalized Bell measurement is an entangling quantum measurement.
For $d=2$, this construction gives the standard Bell measurement.

\appsubsection{Symmetries of generalized Bell measurements}

Let $G$ be the local Weyl group, including phases.  After quotienting out phases,
its action is represented by $W_{a,b}\otimes W_{c,e}$ with
$(a,b),(c,e)\in\bZ_d^2$.  This action maps Bell states to Bell states and is
transitive on $\{\ket{\Psi_{x,z}}:(x,z)\in\bZ_d^2\}$, so the GBM is locally
transitive.

The stabilizer of $\ket{\Psi^\star}$ is, up to phases,
\begin{align}
  G_{\Psi^\star}
  =\{W_{a,b}\otimes W_{a,-b}:(a,b)\in\bZ_d^2\}.
\end{align}
The Bell states are simultaneous eigenvectors of this stabilizer:
\begin{align}
  (W_{a,b}\otimes W_{a,-b})\ket{\Psi_{x,z}}
  = \chi_{a,b}(x,z)\ket{\Psi_{x,z}},
\end{align}
where $\chi_{a,b}(x,z)$ is a phase depending on $(a,b)$ and $(x,z)$.  For prime
$d$, different pairs $(x,z)$ give different characters.  Hence the Bell basis
decomposes into one dimensional irreducible subspaces under the stabilizer, and
the irreducibility condition in Definition~\ref{def:irreducibility} holds.

\appsubsection{Proof of Corollary~\ref{thm:GBM}}

We prove Corollary~\ref{thm:GBM} by constructing a local protocol and showing
that its representative operator is the optimal separable verification operator
for $\vert\Psi^\star\rangle$.

\appsubsubsection{Verification protocol}

Assume that $d$ is prime, so that $\bZ_d$ is the finite field $\bF_d$.  Introduce
the equivalence relation on $\bF_d^2\setminus\{(0,0)\}$ by
$(x,z)\sim (x',z')$ if there exists $a\in\bF_d\setminus\{0\}$ such that
$(ax,az)=(x',z')$.  For a Weyl operator $W$, denote by $\mathscr{B}_W$ an
eigenbasis of $W$.  If $(x,z)\sim(x',z')$, then
$\mathscr{B}_{W_{x,z}}=\mathscr{B}_{W_{x',z'}}$.  A convenient set of
representatives is
\begin{align}
  \mathscr{M} := \{W_{0,1}\equiv Z,\; W_{1,0}\equiv X,\; W_{1,1}\equiv XZ,\;
                    W_{1,2}\equiv XZ^2,\; \cdots,\; W_{1,d-1}\equiv XZ^{d-1}\}.
\end{align}
The eigenbases associated with $\mathscr{M}$ form a complete set of $d+1$
mutually unbiased bases~\cite{DURT2010}.  They satisfy the identity
~\cite{H1,Zhu2019Optimal}
\begin{align}\label{eq:MUB-property}
  \frac{1}{d+1}\sum_{W\in\mathscr{M}}\sum_{\ket{u}\in\mathscr{B}_W}
  \proj{u, \overline{u}}
= \vert\Psi^\star\rangle\!\langle\Psi^\star\vert
  + \frac{1}{d+1}(\1 - \vert\Psi^\star\rangle\!\langle\Psi^\star\vert),
\end{align}
where $\ket{\overline{u}}$ is the complex conjugate of $\ket{u}$ in the
computational basis.

For each $W\in\mathscr{M}$, let
$\mathscr{B}_{\overline W}:=\{\ket{\overline u}:\ket u\in\mathscr{B}_W\}$ be the
conjugate basis.  Define the local input ensemble
\begin{align}
  \cS_{\cB} = \left\{\frac{1}{d^2(d+1)},
                      \left(\proj{u,u'},\cB(\proj{u,u'})\right)
                      \sbar \ket{u}\in\mathscr{B}_W,\ 
                      \ket{u'}\in\mathscr{B}_{\overline W},\
                      W\in\mathscr{M}
              \right\}.
\end{align}
Thus $\cS_{\cB}$ contains $(d+1)d^2$ local test states.  In each round, one
uniformly samples $W\in\mathscr{M}$, then uniformly and independently samples
$\ket{u}\in\mathscr{B}_W$ and $\ket{u'}\in\mathscr{B}_{\overline W}$, prepares
$\proj{u}\otimes\proj{u'}$, and accepts exactly the outcomes in
$\cB(\proj{u,u'})$.

\appsubsubsection{Performance analysis}

By local transitivity, it is enough to compute the verification operator for the
reference outcome $(0,0)$.  For
$\ket{u}\in\mathscr{B}_W$ and $\ket{u'}\in\mathscr{B}_{\overline W}$,
\begin{align}
  \left|\bra{\Psi_{0,0}}u,u'\rangle\right|^2
  = \frac{1}{d}\left|\langle\overline u|u'\rangle\right|^2.
\end{align}
Therefore the local inputs that can produce the ideal outcome $(0,0)$ are exactly
\begin{align}
  \Xi_{0,0} = \bigcup_{W\in\mathscr{M}}
  \left\{\ket{u,\overline{u}}\sbar\ket{u}\in\mathscr{B}_W\right\}.
\end{align}
It follows that
\begin{align}
\Omega_{0,0}(\cS_{\cB})
&= \frac{1}{d^2(d+1)}
\sum_{W\in\mathscr{M}}
\sum_{\ket{u}\in\mathscr{B}_W}\proj{u, \overline{u}} \\
&= \frac{1}{d^2}
\left(\vert\Psi^\star\rangle\!\langle\Psi^\star\vert + \frac{1}{d+1}(\1 -
      \vert\Psi^\star\rangle\!\langle\Psi^\star\vert)\right)
 \label{eq:rdnyiy2}
= \frac{1}{d^2}
\left(
\vert\Psi_{0,0}\rangle\!\langle\Psi_{0,0}\vert + \frac{1}{d+1}(\1 - \vert\Psi_{0,0}\rangle\!\langle\Psi_{0,0}\vert)
\right),
\end{align}
where Eq.~\eqref{eq:rdnyiy2} follows from Eq.~\eqref{eq:MUB-property}.  Equivalently,
the representative operator $T^\star_{\cB}:=d^2\Omega_{0,0}(\cS_{\cB})$ is
\begin{align}
T^\star_{\cB}
= \vert\Psi^\star\rangle\!\langle\Psi^\star\vert
+ \frac{1}{d+1}(\1-\vert\Psi^\star\rangle\!\langle\Psi^\star\vert).
\end{align}
Covariance then gives, for every $(x,z)\in\bZ_d^2$,
\begin{align}\label{eq:GBM-verification-operator-appx}
\Omega_{x,z}(\cS_{\cB})
=\frac{1}{d^2}\left[|\Psi_{x,z}\rangle\langle \Psi_{x,z}| 
+ \frac{1}{d+1}\left(\1-|\Psi_{x,z}\rangle\langle \Psi_{x,z}|\right)\right].
\end{align}
Thus the homogeneous parameter is $\beta=1/(d+1)$, and
Theorem~\ref{thm:verification-efficiency} gives
\begin{align}
P(\cS_{\cB},\varepsilon)
=1-\frac{d}{d+1}\varepsilon.
\end{align}
Substituting this value into Eq.~\eqref{eq:number-of-copies} yields
\begin{align}
N_{\cB}
=\left\lceil
\frac{\ln\delta}{\ln\left(1-\frac{d}{d+1}\varepsilon\right)}
\right\rceil
\end{align}
and, for small $\varepsilon$,
\begin{align}
N_{\cB}
\approx
\left\lceil
\frac{d+1}{d}\cdot\frac{1}{\varepsilon}\ln\frac{1}{\delta}
\right\rceil,
\end{align}
which is the claimed scaling in Corollary~\ref{thm:GBM}.

It remains to justify optimality.  By the symmetry reduction in
Theorem~\ref{thm:verification-efficiency}, any local protocol for the GBM induces a
separable verification operator $T$ for the maximally entangled state
$\ket{\Psi^\star}$.  The separable verification bound in
Ref.~\cite[Theorem 1]{Hayashi_2006} states that the smallest possible second
eigenvalue of such a $T$ is $1/(d+1)$, achieved by $T^\star_{\cB}$ above.
Hence no local GBM verification protocol
can improve on the passing probability or sample complexity above.

\appsection{Single-parameter measurement on two qubits verification}\label{appx:sec:single-parameter}

\appsubsection{Introduction to single-parameter measurements on two qubits}

The single-parameter measurement on two qubits $\cP_\gamma$ is defined in
Eq.~\eqref{eq:single-parameter}.  It interpolates between product basis
measurements at $\gamma=0,\pi/2$ and the Bell measurement at
$\gamma=\pi/4$.  This appendix presents the local protocol used in
Corollary~\ref{thm:single-parameter two-qubit measurement} and computes its
homogeneous verification operators explicitly.

\appsubsection{Symmetries of single-parameter measurements on two qubits}

Direct substitution in Eq.~\eqref{eq:entangled-states} shows that, up to irrelevant phases, the local
unitaries generated by
\begin{align}
  \1\otimes X,\qquad
  Z\otimes Z,\quad
  ZX\otimes \1
\end{align}
permute the four projectors of $\cP_\gamma$.  More explicitly, for
$x,z\in\bZ_2$,
\begin{align}
  (\1\otimes X)\ket{\Psi_{x,z}^{\gamma}}
  &= \ket{\Psi_{x,z+1}^{\gamma}}, \\
  (ZX\otimes\1)\ket{\Psi_{x,z}^{\gamma}}
  &= (-1)^x\ket{\Psi_{x+1,z}^{\gamma}}, \\
  (Z\otimes Z)\ket{\Psi_{x,z}^{\gamma}}
  &= (-1)^{x+z}\ket{\Psi_{x,z}^{\gamma}},
\end{align}
where all additions in the indices are modulo $2$.  Equivalently,
\begin{align}
  (\1\otimes X)\proj{\Psi_{x,z}^{\gamma}}(\1\otimes X)^\dagger
  &= \proj{\Psi_{x,z+1}^{\gamma}}, \\
  (ZX\otimes\1)\proj{\Psi_{x,z}^{\gamma}}(ZX\otimes\1)^\dagger
  &= \proj{\Psi_{x+1,z}^{\gamma}}, \\
  (Z\otimes Z)\proj{\Psi_{x,z}^{\gamma}}(Z\otimes Z)^\dagger
  &= \proj{\Psi_{x,z}^{\gamma}} .
\end{align}
Thus $\1\otimes X$ exchanges the $z$ label, while $ZX\otimes\1$ exchanges
the $x$ label.  Hence the family is locally transitive, and a symmetric local
protocol can be specified by computing one representative verification
operator and transporting it to the remaining outcomes.

\appsubsection{Proof of Corollary~\ref{thm:single-parameter two-qubit measurement}}

We prove Corollary~\ref{thm:single-parameter two-qubit measurement} by
constructing the local protocol and computing its verification operators.

\appsubsubsection{Verification protocol}

Define
\begin{subequations}\label{eq:post-measurement-states}
\begin{align}
  \ket{v_\pm} &= \sin\gamma\ket{0} \pm \cos\gamma\ket{1}, \\
  \ket{v_\pm^\perp} &= \cos\gamma\ket{0} \mp \sin\gamma\ket{1}, \\
  \ket{w_\pm} &= \sin\gamma\ket{0} \pm i\cos\gamma\ket{1}, \\
  \vert{w_\pm^\perp} \rangle &= \cos\gamma\ket{0} \pm i\sin\gamma\ket{1}.
\end{align}
\end{subequations}
The superscript $\perp$ denotes the orthogonal state in the corresponding
two dimensional basis.

The protocol $\cS_\gamma$ is listed in
Table~\ref{tbl:projective-measurement-verification}.  In each round, one
samples a row uniformly, prepares the corresponding local product state, and
accepts precisely the ideal outcomes listed in the last column.

\begin{table}[!hbtp]
\centering
\setlength{\tabcolsep}{10pt}
\renewcommand{\arraystretch}{1.3}
\begin{tabular}{@{}cccc@{}}
    \toprule
      \textbf{Index $x$} &
      \textbf{Prob. $p_x$} & \textbf{Test State $\rho_x$} & \textbf{Outputs $\cP_\gamma(\rho_x)$} \\\midrule
      1 & $1/12$ & $\ket{0}\otimes\ket{0}$ & $\{(0,0),\;(1,1)\}$ \\
      2 & $1/12$ & $\ket{1}\otimes\ket{1}$ & $\{(0,0),\;(1,1)\}$ \\
      3 & $1/12$ & $\ket{0}\otimes\ket{1}$ & $\{(0,1),\;(1,0)\}$ \\
      4 & $1/12$ & $\ket{1}\otimes\ket{0}$ & $\{(0,1),\;(1,0)\}$ \\
      5 & $1/12$ & $\ket{v_+}\otimes\ket{+}$ & $\{(0,0),\;(0,1)\}$ \\
      6 & $1/12$ & $\ket{v_-}\otimes\ket{-}$ & $\{(0,0),\;(0,1)\}$ \\
      7 & $1/12$ & $\ket{v_+^\perp}\otimes\ket{+}$ & $\{(1,0),\;(1,1)\}$ \\
      8 & $1/12$ & $\ket{v_-^\perp}\otimes\ket{-}$ & $\{(1,0),\;(1,1)\}$ \\
      9 & $1/12$ & $\ket{\top}\otimes\ket{w_-}$    & $\{(0,0),\;(1,0)\}$ \\
      10 & $1/12$ & $\ket{\perp}\otimes\ket{w_+}$  & $\{(0,0),\;(1,0)\}$ \\
      11 & $1/12$ & $\ket{\top}\otimes\ket{w_+^\perp}$ & $\{(0,1),\;(1,1)\}$ \\
      12 & $1/12$ & $\ket{\perp}\otimes\ket{w_-^\perp}$ & $\{(0,1),\;(1,1)\}$ \\
    \bottomrule
\end{tabular}
\caption{A homogeneous local verification protocol for the single-parameter
measurement on two qubits $\cP_\gamma$.  Here
$\{\ket{0},\ket{1}\}$, $\{\ket{+},\ket{-}\}$, and
$\{\ket{\top},\ket{\perp}\}$ are the eigenstates of the Pauli operators
$Z$, $X$, and $Y$, respectively, with
$\ket{\top}=(\ket0+i\ket1)/\sqrt2$ and
$\ket{\perp}=(\ket0-i\ket1)/\sqrt2$.}
\label{tbl:projective-measurement-verification}
\end{table}

\appsubsubsection{Performance analysis}

For the reference outcome $(0,0)$, the accepting states in
Table~\ref{tbl:projective-measurement-verification} give
\begin{align}
\Omega_{0,0}(\cS_\gamma)
&= \frac{1}{12}\left(\proj{00} + \proj{11}
 + \proj{v_++} + \proj{v_--}
 + \proj{\top w_-} + \proj{\perp w_+}\right) \\
&= \frac{1}{4}\left[
\proj{\Psi_{0,0}^{\gamma}}
+\frac{1}{3}\left(\1-\proj{\Psi_{0,0}^{\gamma}}\right)
\right].
\end{align}
The local symmetry generated by $ZX\otimes\1$, $\1\otimes X$, and
$Z\otimes Z$ maps this calculation to the remaining three outcomes.  Therefore
\begin{align}
\Omega_{x,z}(\cS_\gamma)
=\frac{1}{4}\left[
\proj{\Psi_{x,z}^{\gamma}}
+\frac{1}{3}\left(\1-\proj{\Psi_{x,z}^{\gamma}}\right)
\right],
\quad (x,z)\in\bZ_2^2.
\end{align}
The protocol is homogeneous with $\beta=1/3$.  For any POVM
$\cM=\{M_{x,z}\}_{(x,z)\in\bZ_2^2}$,
\begin{align}
\sum_{x,z}\tr[\Omega_{x,z}(\cS_\gamma)M_{x,z}]
=\beta+(1-\beta)\Fid(\cP_\gamma,\cM).
\end{align}
Therefore, under the bad case condition
$\Fid(\cP_\gamma,\cM)\leq1-\varepsilon$,
\begin{align}
P(\cS_\gamma,\varepsilon)=1-\frac{2}{3}\varepsilon.
\end{align}
The equality is tight, for example by mixing the ideal measurement with any
fixed point free relabeling of its four outcomes.
Substitution into Eq.~\eqref{eq:number-of-copies-2} yields
\begin{align}
N_{\cP_\gamma}
=
\left\lceil
\frac{\ln\delta}{\ln P(\cS_\gamma,\varepsilon)}
\right\rceil
\approx
\left\lceil
\frac{3}{2}\cdot\frac{1}{\varepsilon}\ln\frac{1}{\delta}
\right\rceil,
\end{align}
as stated in Corollary~\ref{thm:single-parameter two-qubit measurement}.
At $\gamma=\pi/4$ the protocol reduces to the optimal Bell measurement protocol.
At $\gamma=0$ and $\gamma=\pi/2$, a direct product basis verification reaches the
unconstrained optimal scaling.  For generic $\gamma$, the corollary records the
performance of this homogeneous symmetric local protocol rather than a separate
global optimality claim.

\appsection{Elegant joint measurement verification}\label{appx:sec:elegant-joint-measurement}

\appsubsection{Introduction to elegant joint measurements}

The EJM family on two qubits $\cJ_\kappa$ is defined in Eq.~\eqref{eq:EJM}.  Each
state has tetrahedral one qubit marginals, with a Bloch length that depends on
$\kappa$.  To describe the local test states used below, let $\ket{\bm r}$
denote the pure state of one qubit whose Bloch vector is the unit vector
$\bm r$.
Define the four tetrahedral directions
\begin{align}
\bm t_0&=\frac{1}{\sqrt{3}}(1,1,1),\\
\bm t_1&=\frac{1}{\sqrt{3}}(1,-1,-1), \\
\bm t_2&=\frac{1}{\sqrt{3}}(-1,1,-1),\\
\bm t_3&=\frac{1}{\sqrt{3}}(-1,-1,1).
\end{align}
For the basis in Eq.~\eqref{eq:EJM-states},
\begin{align}
\tr_B\proj{\Phi_j^\kappa}
=\frac{1}{2}\left(\1+\frac{\sqrt{3}}{2}\cos\kappa\,
\bm t_j\cdot\bm\sigma\right), \qquad
\tr_A\proj{\Phi_j^\kappa}
=\frac{1}{2}\left(\1-\frac{\sqrt{3}}{2}\cos\kappa\,
\bm t_j\cdot\bm\sigma\right).
\end{align}
Thus the reduced states point along opposite tetrahedral directions, except at
the maximally entangled endpoint $\kappa=\pi/2$, where both marginals are
maximally mixed.

\appsubsection{Symmetries of elegant joint measurements}

Direct substitution in Eq.~\eqref{eq:EJM-states} gives
\begin{align}
(X\otimes X)\ket{\Phi_0^\kappa}
&=-\ket{\Phi_1^\kappa},&
(X\otimes X)\ket{\Phi_1^\kappa}
&=-\ket{\Phi_0^\kappa}, \\
(X\otimes X)\ket{\Phi_2^\kappa}
&=-\ket{\Phi_3^\kappa},&
(X\otimes X)\ket{\Phi_3^\kappa}
&=-\ket{\Phi_2^\kappa},
\end{align}
and
\begin{align}
(Z\otimes Z)\ket{\Phi_0^\kappa}
&=-\ket{\Phi_3^\kappa},&
(Z\otimes Z)\ket{\Phi_3^\kappa}
&=-\ket{\Phi_0^\kappa}, \\
(Z\otimes Z)\ket{\Phi_1^\kappa}
&=-\ket{\Phi_2^\kappa},&
(Z\otimes Z)\ket{\Phi_2^\kappa}
&=-\ket{\Phi_1^\kappa}.
\end{align}
Therefore, at the level of projectors, $X\otimes X$ induces the
permutation $(01)(23)$ and $Z\otimes Z$ induces the permutation
$(03)(12)$. The group generated by these two permutations is transitive
on $\{0,1,2,3\}$. Hence $\cJ_\kappa$ is locally transitive for every
$\kappa\in[0,\pi/2]$. The endpoint $\kappa=\pi/2$ is included because
the above identities hold without any division by a $\kappa$-dependent
quantity.
The one qubit tetrahedral
rotation $U$ satisfying
\begin{align}
  U X U^\dagger=Y,\qquad UYU^\dagger=Z,\qquad UZU^\dagger=X
\end{align}
also gives a local symmetry: $U\otimes U$ fixes the projector
$\proj{\Phi_0^\kappa}$ and cycles the remaining three projectors.  Hence the
verification operator for one representative outcome determines the operators
for the other outcomes by local covariance.

\appsubsection{Proof of Corollary~\ref{thm:elegant-joint-measurement}}

We prove Corollary~\ref{thm:elegant-joint-measurement} by constructing a local
protocol and evaluating its homogeneous verification operators.

\appsubsubsection{Verification protocol}

The protocol $\cS_\kappa$ consists of three symmetry orbits of local product
tests.  
First, for each unordered pair $0\le j<\ell\le3$, consider vectors in
the two dimensional subspace
$\operatorname{span}\{\ket{\Phi_j^\kappa},\ket{\Phi_\ell^\kappa}\}$ of
the form
\begin{align}
    \ket{\Phi_j^\kappa}+z\ket{\Phi_\ell^\kappa}.
\end{align}
Such a vector is a product vector if and only if the determinant of its
$2\times2$ coefficient matrix in the computational product basis
vanishes.  Direct substitution of Eq.~\eqref{eq:EJM-states} shows that,
for every unordered pair $\{j,\ell\}$, this determinant condition is
equivalent, up to an irrelevant nonzero scalar factor, to
\begin{align}\label{eq:EJM-product-root-equation}
  (e^{2i\kappa}-3)z^2
  +2(e^{2i\kappa}+1)z
  +(e^{2i\kappa}-3)=0 .
\end{align}
Let $z_\pm^\kappa$ denote the two roots of this equation, counted with
multiplicity, and define
\begin{align}\label{eq:EJM-pair-roots}
  \ket{\eta_{j\ell,\pm}^\kappa}
  =
  \frac{\ket{\Phi_j^\kappa}+z_\pm^\kappa\ket{\Phi_\ell^\kappa}}
       {\sqrt{1+|z_\pm^\kappa|^2}} .
\end{align}
Since $\ket{\eta_{j\ell,\pm}^\kappa}$ lies in
$\operatorname{span}\{\ket{\Phi_j^\kappa},\ket{\Phi_\ell^\kappa}\}$,
it is automatically orthogonal to the other two EJM basis states.
Therefore its accepted outcomes are precisely $\{j,\ell\}$.

For $0<\kappa\leq\pi/2$, the two roots are distinct.  At
$\kappa=0$, Eq.~\eqref{eq:EJM-product-root-equation} becomes
\begin{align}
    -2(z-1)^2=0,
\end{align}
so the two roots coincide.  In this degenerate endpoint, the two labels
$\eta_{j\ell,+}^0$ and $\eta_{j\ell,-}^0$ represent the same product
test.  The protocol is then understood with the duplicate tests merged,
or equivalently with their probabilities added.

Set
\begin{align}\label{eq:EJM-orbit-weights}
  a_\kappa=\frac{1+\sin\kappa}{2(2-\sin\kappa)},\qquad
  b_\kappa=\frac{1-\sin\kappa}{2-\sin\kappa},\qquad
  c_\kappa=\frac{1-\sin\kappa}{2(2-\sin\kappa)} .
\end{align}
The protocol is:
\begin{itemize}
\item prepare each $\ket{\eta_{j\ell,\pm}^\kappa}$ with probability
$a_\kappa/12$ and accept the outcomes $\{j,\ell\}$;
\item prepare each $\ket{\bm t_j,\bm t_j}$ and
$\ket{-\bm t_j,-\bm t_j}$ with probability $b_\kappa/8$ and accept all outcomes
except $j$;
\item prepare each $\ket{\bm t_j,-\bm t_j}$ with probability $c_\kappa/4$ and
accept all four outcomes.
\end{itemize}
The probabilities sum to one because $a_\kappa+b_\kappa+c_\kappa=1$.  At
$\kappa=0$, the two roots in Eq.~\eqref{eq:EJM-product-root-equation} coincide,
and the first orbit reduces to the six Pauli tests with opposite local
directions used for the original EJM.

\appsubsubsection{Performance analysis}

Let $P_j^\kappa:=\proj{\Phi_j^\kappa}$.  For the reference outcome $0$, define
\begin{align}
  R_0^\kappa&:=\sum_{\ell=1}^{3}
  \left(\ketbra{\Phi_0^\kappa}{\Phi_\ell^\kappa}
       +\ketbra{\Phi_\ell^\kappa}{\Phi_0^\kappa}\right),\\
  Q_0^\kappa&:=\sum_{1\le \ell<m\le3}
  \left(\ketbra{\Phi_\ell^\kappa}{\Phi_m^\kappa}
       +\ketbra{\Phi_m^\kappa}{\Phi_\ell^\kappa}\right).
\end{align}
Direct substitution of Eq.~\eqref{eq:EJM-states}, together with
$\proj{\bm r}=(\1+\bm r\cdot\bm\sigma)/2$, gives the three orbit contributions
accepted by outcome $0$:
\begin{align}
A_0^\kappa
&:=\frac{1}{12}\sum_{\ell=1}^{3}\sum_{\nu=\pm}
\proj{\eta_{0\ell,\nu}^\kappa} 
=
\frac{1}{4}P_0^\kappa+\frac{1}{12}\sum_{\ell=1}^{3}P_\ell^\kappa
+\frac{1-\sin\kappa}{12(1+\sin\kappa)}R_0^\kappa, \label{eq:EJM-A0}\\
B_0^\kappa
&:=\frac{1}{8}\sum_{\ell=1}^{3}
\left(\proj{\bm t_\ell,\bm t_\ell}
      +\proj{-\bm t_\ell,-\bm t_\ell}\right) 
=
\frac{1}{4}P_0^\kappa+\frac{1}{6}\sum_{\ell=1}^{3}P_\ell^\kappa
-\frac{1}{12}R_0^\kappa-\frac{1}{24}Q_0^\kappa, \label{eq:EJM-B0}\\
C_0^\kappa
&:=\frac{1}{4}\sum_{\ell=0}^{3}\proj{\bm t_\ell,-\bm t_\ell} 
=
\frac{1}{4}\sum_{\ell=0}^{3}P_\ell^\kappa
+\frac{1}{12}\left(R_0^\kappa+Q_0^\kappa\right). \label{eq:EJM-C0}
\end{align}
The choices in Eq.~\eqref{eq:EJM-orbit-weights} satisfy
$b_\kappa=2c_\kappa$ and
\begin{align}
a_\kappa\frac{1-\sin\kappa}{12(1+\sin\kappa)}
=\frac{b_\kappa}{24}.
\end{align}
Therefore the terms outside the diagonal in
$a_\kappa A_0^\kappa+b_\kappa B_0^\kappa+c_\kappa C_0^\kappa$ cancel.  The
verification operator for the reference outcome is
\begin{align}
\Omega_0(\cS_\kappa)
&=a_\kappa A_0^\kappa+b_\kappa B_0^\kappa+c_\kappa C_0^\kappa
=\frac{1}{4}\left[P_0^\kappa+\beta_\kappa(\1-P_0^\kappa)\right],
\end{align}
where
\begin{align}
  \beta_\kappa
  =\frac{a_\kappa}{3}+\frac{2b_\kappa}{3}+c_\kappa
  =\frac{4-3\sin\kappa}{3(2-\sin\kappa)}.
\end{align}

By local symmetry,
\begin{align}
\Omega_j(\cS_\kappa)
=\frac{1}{4}\left[
P_j^\kappa+\beta_\kappa(\1-P_j^\kappa)
\right],
\qquad j=0,1,2,3.
\end{align}
Thus the EJM protocol is homogeneous.  Substitution into
the direct homogeneous identity
\begin{align}
\sum_j\tr[\Omega_j(\cS_\kappa)M_j]
=\beta_\kappa+(1-\beta_\kappa)\Fid(\cJ_\kappa,\cM)
\end{align}
gives
\begin{align}
P(\cS_\kappa,\varepsilon)
=1-(1-\beta_\kappa)\varepsilon
=1-\frac{2}{3(2-\sin\kappa)}\varepsilon.
\end{align}
Again, the equality is tight by mixing the ideal measurement with a fixed point
free relabeling of the four EJM outcomes.

Substitution into Eq.~\eqref{eq:number-of-copies-2} yields
\begin{align}
N_{\cJ_\kappa}
=
\left\lceil
\frac{\ln\delta}{\ln P(\cS_\kappa,\varepsilon)}
\right\rceil
\approx
\left\lceil
\frac{3(2-\sin\kappa)}{2}\cdot
\frac{1}{\varepsilon}\ln\frac{1}{\delta}
\right\rceil,
\end{align}
as stated in Corollary~\ref{thm:elegant-joint-measurement}.

\appsection{Stabilizer state induced measurement verification}\label{appx:sec:stabilizer-state}

\appsubsection{Introduction to stabilizer state induced measurements}

The stabilizer state induced measurement $\cM_S$ is defined in
Eq.~\eqref{eq:stabilizer state induced measurements}.  It is obtained from the
orbit of a stabilizer state $\ket{S}$ under the $n$ qudit Weyl representation.
The protocol below is adapted from stabilizer state
verification~\cite[Section X]{Zhu2019General}.

\appsubsection{Symmetries of stabilizer state induced measurements}

The Weyl operators $W^n(\vec{x})$ are tensor products of local Weyl operators.
They map
\begin{align}
  \ket{S,\vec{a}}\mapsto W^n(\vec{x})\ket{S,\vec{a}}
  = \ket{S,\vec{x}+\vec{a}},
\end{align}
up to phases and modulo the stabilizer subgroup $N$.  Hence the action is
locally transitive on $\bF_d^{2n}/N$.

The stabilizer of the reference outcome is $N$.  
Each basis
state $\ket{S,\vec a}=W^n(\vec a)\ket S$ is a simultaneous eigenvector of
the stabilizer operators.  More explicitly, the relation
\begin{align}
  W^n(\vec n)\ket{S,\vec a}
  =
  \chi_{\vec a}(\vec n)\ket{S,\vec a},
  \qquad
  \chi_{\vec a}(\vec n)=\omega^{[\vec n,\vec a]}
\end{align}
holds for $\vec n\in N$.  If two
cosets $\vec a+N$ and $\vec b+N$ define the same character on $N$, then
\begin{align}
    [\vec n,\vec a-\vec b]=0 \qquad \forall\,\vec n\in N.
\end{align}
Hence $\vec a-\vec b\in N^\perp$.  Since $N$ is maximal isotropic,
$N^\perp=N$, and therefore $\vec a+N=\vec b+N$.  Thus different cosets in
$\bF_d^{2n}/N$ define inequivalent one-dimensional characters of the
stabilizer.  Consequently, the stabilizer representation on the basis
$\{\ket{S,\vec a}\}$ is multiplicity free, and the irreducibility
condition in Definition~\ref{def:irreducibility} is satisfied.

\appsubsection{Proof of Corollary~\ref{thm:stabilizer state induced measurements}}

We prove Corollary~\ref{thm:stabilizer state induced measurements} by
constructing a local stabilizer test protocol and evaluating the associated
homogeneous verification operator.

\appsubsubsection{Verification protocol}

Choose stabilizer generators $\{W^n(\vec{x}_r):r=1,\ldots,n\}$ for $\ket{S}$.
For $\vec{k}=(k_1,\ldots,k_n)\in\bZ_d^n$, define
\begin{align}
  S_{\vec{k}}
  :=\prod_{r=1}^{n} W^n(\vec{x}_r)^{k_r}.
\end{align}
There are $d^n$ such stabilizers, and $\vec{k}=\vec{0}$ gives the identity.
Let $\mathscr{B}_{\vec{k}}$ be a local product eigenbasis of $S_{\vec{k}}$, and
let $\mathscr{B}^1_{\vec{k}}$ be the subset with eigenvalue $1$.  Then
$|\mathscr{B}_{\vec{k}}|=d^n$ and $|\mathscr{B}^1_{\vec{k}}|=d^{n-1}$ for
$\vec{k}\neq\vec{0}$. The stabilizer state identity~\cite{Zhu2019General}
\begin{align}\label{eq:stabilizer-identity}
\frac{1}{d^n-1}
\sum_{\vec{k}\in\bZ_d^n\backslash\{\vec{0}\}} \;
\sum_{\ket{\vec{u}}\in\mathscr{B}^1_{\vec{k}}}\proj{\vec{u}}
=\proj{S} + \frac{d^{n-1}-1}{d^n-1}(\1 - \proj{S})
\end{align}
is used below.

Define the local protocol
\begin{align}\label{eq:protocol-for-stabilizer-measurement}
  \cS_{\cM_S} = \left\{\frac{1}{d^n(d^n-1)},
                      \left(\proj{\vec{u}},\cM_S(\proj{\vec{u}})\right)
                      \sbar \ket{\vec{u}}\in\mathscr{B}_{\vec{k}},\; \vec{k}\in\bZ_d^n\backslash\{\vec{0}\}
              \right\}.
\end{align}
That is, one samples a nonidentity stabilizer uniformly, samples a local product
eigenstate of that stabilizer uniformly, and accepts exactly the ideal outcomes
that can occur for this input state.

\appsubsubsection{Performance analysis}

Fix an outcome $\vec{a}\in\bF_d^{2n}/N$ and define
\begin{align}
    \mathscr{B}_{\vec{k}}(\vec{a})
:=  \left\{W^n(\vec{a})\ket{\vec{u}} \sbar
          \ket{\vec{u}}\in\mathscr{B}^1_{\vec{k}}\right\}.
\end{align}
This set is exactly the set of test states from the $\vec{k}$-basis that
are accepted by the outcome $\vec a$.  To see this, first consider the
reference outcome $\vec 0$.  Since $\ket S$ is stabilized by every
nonidentity stabilizer $S_{\vec{k}}$, it lies in the eigenvalue-one
eigenspace of $S_{\vec{k}}$.  Hence, in the eigenbasis
$\mathscr{B}_{\vec{k}}$, only the vectors in
$\mathscr{B}_{\vec{k}}^1$ can have nonzero overlap with $\ket S$.
Therefore the states in $\mathscr{B}_{\vec{k}}^1$ are precisely the
states in $\mathscr{B}_{\vec{k}}$ accepted by the reference outcome.

Conjugation by $W^n(\vec{a})$ maps eigenstates of $S_{\vec{k}}$ to eigenstates
of the same stabilizer, possibly with a shifted eigenvalue.  Hence
$\mathscr{B}_{\vec{k}}(\vec{a})$ is a subset of the test states in
Eq.~\eqref{eq:protocol-for-stabilizer-measurement}.  These are exactly the test
states selected by the protocol that have support on the ideal outcome
$\vec{a}$.  Therefore
\begin{align}
\Omega_{\vec{a}}(\cS_{\cM_S})
&= \frac{1}{d^n(d^n-1)}
\sum_{\vec{k}\in\bZ_d^n\backslash\{\vec{0}\}}
                \sum_{\ket{\vec{u}}\in\mathscr{B}_{\vec{k}}^1}
                W^n(\vec{a})\proj{\vec{u}}(W^n(\vec{a}))^\dagger \\
&= \frac{1}{d^n}W^n(\vec{a})
      \left(\frac{1}{d^n-1}\sum_{\vec{k}\in\bZ_d^n\backslash\{\vec{0}\}}
                \sum_{\ket{\vec{u}}\in\mathscr{B}_{\vec{k}}^1}
                    \proj{\vec{u}}\right)
                (W^n(\vec{a}))^\dagger \\
&= \frac{1}{d^n}W^n(\vec{a})
                \left(\proj{S} + \frac{d^{n-1}-1}{d^n-1}(\1 - \proj{S})\right)
                (W^n(\vec{a}))^\dagger \\
&= \frac{1}{d^n}\left[
\proj{S,\vec{a}} + \frac{d^{n-1}-1}{d^n-1}
\left(\1 - \proj{S,\vec{a}}\right)\right].
\end{align}
Thus $\beta=(d^{n-1}-1)/(d^n-1)$.  Substitution into
Eq.~\eqref{eq:homogeneous-failure-probability} gives
\begin{align}
P(\cS_{\cM_S},\varepsilon)
=1-\frac{d^n-d^{n-1}}{d^n-1}\varepsilon.
\end{align}
Substitution into Eq.~\eqref{eq:number-of-copies-2} gives
\begin{align}
N_{\cM_S}
=
\left\lceil
\frac{\ln\delta}{\ln P(\cS_{\cM_S},\varepsilon)}
\right\rceil
\approx
\left\lceil
\frac{d^n - 1}{d^n - d^{n-1}}\cdot
\frac{1}{\varepsilon}\ln\frac{1}{\delta}
\right\rceil,
\end{align}
which gives the asymptotic sample complexity in
Corollary~\ref{thm:stabilizer state induced measurements}.  The exact formula is
the ceiling expression above.  The protocol is
the standard homogeneous stabilizer test strategy.  Theorem~\ref{thm:verification-efficiency}
reduces the relevant optimality question to the second eigenvalue of the
representative stabilizer state verification operator; within the stabilizer test
construction above, this eigenvalue is fixed by the stabilizer identity in
Eq.~\eqref{eq:stabilizer-identity}, so no different weighting of these
nonidentity stabilizer eigenspaces improves the homogeneous parameter.

\appsection{Discussion on measurement fidelity estimation}\label{appx:measurement-fidelity-estimation}

This appendix provides the details and proofs for measurement fidelity estimation based on QMV.
Assume that a homogeneous verification
protocol $\cS=\{\Omega_\theta:\theta\in\Theta\}$ for the target quantum measurement $\cP$
has been constructed, where each $\Omega_\theta$ has the form of
Eq.~\eqref{eq:homogeneous}.

To use this homogeneous QMV strategy to achieve fidelity estimation,
the protocol is performed for $N$ rounds.
In round $i$, a test state is sampled and fed into the measurement device.
Based on the outcome, define the random variable $X_i$ by
\begin{align}
    X_i =   \begin{cases}
                1, &\text{The state passes the test}, \\
                0, &\text{otherwise}.
            \end{cases}
\end{align}
By construction, each $X_i$ is independent and has a Bernoulli distribution $B(p)$ with
\begin{align}
    p := \sum_{\theta\in\Theta}\tr[\Omega_\theta M_\theta]
       = \sum_{\theta\in\Theta}\tr\left[
          \left(\frac{\proj{\psi_\theta} + \beta(\1 - \proj{\psi_\theta})}{\vert\Theta\vert}\right) M_\theta\right]
       = (1-\beta)\Fid + \beta,\label{eq:yiwre}
\end{align}
where $\Fid$ is the fidelity between the target measurement $\cP$ and the actual measurement $\cM$.
Accordingly, the variance of $X_i$ is $\mathbb{V}[X_i]=p(1-p)$,
where $\mathbb{V}[\cdot]$ denotes variance.
By the law of large numbers, $\widehat{p}:=\frac{1}{N}\sum_{i=1}^NX_i$
approximates the ideal target $p$ as $N$ becomes large, and
$\mathbb{E}[\widehat{p}] = p$,
where $\mathbb{E}[\cdot]$ denotes the expectation value.
From $\wh{p}$, the estimate $\wh{F}$ is obtained directly from Eq.~\eqref{eq:yiwre}.

It remains to analyze how quickly the estimate $\wh{F}$ converges.  First,
\begin{align}
    \mathbb{V}[\widehat{p}]
= \mathbb{V}\left[\frac{1}{N}\sum_{i=1}^NX_i\right]
= \frac{1}{N^2}\sum_{i=1}^N\mathbb{V}[X_i]
= \frac{p(1-p)}{N},
\end{align}
where independence of the rounds is used in the third equality.  Using Eq.~\eqref{eq:yiwre}, we obtain
\begin{align}
    \mathbb{V}[\widehat{F}]
= \frac{\mathbb{V}[\widehat{p}]}{(1-\beta)^2}
= \frac{p(1-p)}{(1-\beta)^2N}.
\end{align}
Thus, the sample standard deviation of the estimated fidelity $\wh{F}$ satisfies
\begin{align}\label{eq:standard-deviation}
    \Delta(\widehat{F})
:= \sqrt{\mathbb{V}[\widehat{F}] } = \frac{\sqrt{p(1-p)}}{(1-\beta)\sqrt{N}}
 = \frac{\sqrt{(1-F)(F + (1-\beta)^{-1} - 1)}}{\sqrt{N}}
 \leq \frac{1}{2(1-\beta)\sqrt{N}},
\end{align}
where the second equality follows from Eq.~\eqref{eq:yiwre}
and the last inequality follows from the arithmetic-mean--geometric-mean inequality.

This bound gives the sample complexity of fidelity estimation.  Since
$\widehat{F}=(\widehat{p}-\beta)/(1-\beta)$ and
$\mathbb{E}[\widehat{p}]=p$, the estimator $\widehat{F}$ is unbiased:
$\mathbb{E}[\widehat{F}]=F$.  If the estimation precision is quantified by the
root-mean-square error and denoted by $\varepsilon_{\rm est}$, then it suffices
to require
\begin{align}
  \Delta(\widehat{F}) \leq \varepsilon_{\rm est}.
\end{align}
Using Eq.~\eqref{eq:standard-deviation}, this condition is guaranteed whenever
\begin{align}
  N \geq
  \left\lceil
  \frac{1}{4(1-\beta)^2\varepsilon_{\rm est}^2}
  \right\rceil.
\end{align}
Therefore the number of rounds needed for fidelity estimation scales as
\begin{align}
  N_{\rm est}=\cO\left(\frac{1}{(1-\beta)^2\varepsilon_{\rm est}^2}\right),
\end{align}
and, for a fixed homogeneous parameter $\beta<1$,
$N_{\rm est}=\cO(1/\varepsilon_{\rm est}^2)$.  Equivalently, Chebyshev's
inequality gives the confidence bound
\begin{align}
\Pr\!\left[|\widehat{F}-F|\geq\varepsilon_{\rm est}\right]
\leq
\frac{1}{4(1-\beta)^2N\varepsilon_{\rm est}^2},
\end{align}
so confidence $1-\delta$ is ensured by
\begin{align}
  N \geq
  \left\lceil
  \frac{1}{4(1-\beta)^2\delta\,\varepsilon_{\rm est}^2}
  \right\rceil.
\end{align}

\end{document}